\documentclass[conference]{IEEEtran}
\usepackage{graphicx}
\usepackage{physics}
\usepackage{amsmath}
\usepackage{amssymb}
\usepackage{amsthm}
\usepackage{enumitem}
\usepackage{framed}
\usepackage[caption=false]{subfig}

\usepackage{adjustbox}
\usepackage{multirow}
\usepackage{multicol}
\usepackage{listings}
\usepackage{tikz}
\usepackage{xcolor}
\usepackage{bm}
\usepackage{mathtools}
\usepackage{algorithm}
\usepackage{algpseudocode,eqparbox}
\usepackage[T1]{fontenc}
\usepackage{wrapfig}
\usetikzlibrary{arrows,decorations.pathmorphing,decorations.footprints,decorations.pathreplacing,fadings,calc,trees,mindmap,shadows,decorations.text,patterns,positioning,shapes,matrix,fit}
\usetikzlibrary{shapes.misc}

\mathtoolsset{showonlyrefs}
\allowdisplaybreaks
\usepackage[colorlinks,citecolor=red,urlcolor=blue,bookmarks=false,hypertexnames=true]{hyperref}

\usepackage{letltxmacro}
\LetLtxMacro{\originaleqref}{\eqref}
\renewcommand{\eqref}{Eq.~\originaleqref}

\definecolor{codegreen}{rgb}{0,0.6,0}
\definecolor{codegray}{rgb}{0.5,0.5,0.5}
\definecolor{codepurple}{rgb}{0.58,0,0.82}
\definecolor{backcolour}{rgb}{0.95,0.95,0.92}

\usepackage{color}
\definecolor{bitcolor}{rgb}{0.99,0.93,0.0}
\definecolor{checkcolor}{rgb}{0.52941,0.80784,1}
\definecolor{channelcolor}{rgb}{0.67,0.88,0.69}
\definecolor{infocolor}{rgb}{0.82,0.62,0.91}
\definecolor{cqcolor}{rgb}{0.99,0.56,0.67}
\lstdefinestyle{mystyle}{
    backgroundcolor=\color{backcolour},   
    commentstyle=\color{codegreen},
    keywordstyle=\color{magenta},
    numberstyle=\tiny\color{codegray},
    stringstyle=\color{codepurple},
    basicstyle=\ttfamily\footnotesize,
    breakatwhitespace=false,         
    breaklines=true,                 
    captionpos=b,                    
    keepspaces=true,                 
    numbers=left,                    
    numbersep=5pt,                  
    showspaces=false,                
    showstringspaces=false,
    showtabs=false,                  
    tabsize=2
}

\usepackage{stmaryrd}

\usepackage[style=ieee,backend=biber,url=false,doi=false,eprint=false]{biblatex}

\bibliography{ref}

\usepackage{ifthen}

\newif\ifarxiv
\arxivtrue

\newif\ifisit
\isitfalse

\newif\ifextra
\extrafalse

\makeatletter
\newcounter{IEEE@bibentries}
\renewcommand\IEEEtriggeratref[1]{%
  \renewbibmacro{finentry}{%
    \stepcounter{IEEE@bibentries}%
    \ifthenelse{\equal{\value{IEEE@bibentries}}{#1}}
    {\finentry\@IEEEtriggercmd}
    {\finentry}%
  }%
}
\makeatother

\usepackage{mathtools}

\usetikzlibrary{shapes.geometric, arrows}
\tikzstyle{startstop} = [rectangle,  minimum width=3cm, minimum height=1cm,text centered,text width=10cm, draw=black ,fill=gray!20]
\tikzstyle{process} = [rectangle, minimum width=3cm, minimum height=1cm, text centered,text width=10cm, draw=black,fill=orange!20]
\tikzstyle{arrow} = [thick,->,>=stealth]
\tikzstyle{state}=[shape=circle,draw=blue!50,fill=blue!20]
\tikzstyle{observation}=[shape=rectangle,draw=orange!50,fill=orange!20]
\tikzstyle{lightedge}=[<-,dotted]
\tikzstyle{mainstate}=[state,thick]
\tikzstyle{mainedge}=[<-,thick]
\usepackage{color}
\definecolor{bitcolor}{rgb}{1,0.84314,0}
\definecolor{checkcolor}{rgb}{0.52941,0.80784,1}
\usepackage{pgfplots}
\pgfplotsset{compat=1.18}
\renewcommand{\epsilon}{\varepsilon}

\newtheorem{theorem}{Theorem}

\newtheorem{rem}[theorem]{Remark}
\newtheorem{lem}[theorem]{Lemma}
\newtheorem{defn}[theorem]{Definition}

\newcommand{\cA}{\mathcal{A}}

\newcommand{\cG}{\mathcal{G}}

\newcommand{\cU}{\mathcal{U}}

\newcommand{\cX}{\mathcal{X}}

\newcommand{\bw}{\bm{w}}
\newcommand{\bx}{\bm{x}}
\newcommand{\by}{\bm{y}}

\newcommand{\mI}{\mathbb{I}}

\newcommand{\hn}{\mathbb{C}^q}
\newcommand{\hop}{\mathbb{S}}
\newcommand{\psd}{\mathbb{S}_{+}}
\newcommand{\pd}{\mathbb{S}_{++}}
\newcommand{\dop}{\mathbb{D}}

\renewcommand{\triangleq}{\coloneqq}

\newcommand{\sym}[1]{S_N}

\newcommand{\vnop}{\varoast}

\newcommand{\cnop}{\boxast}

\newcommand{\blambda}{\bm{\lambda}}
\newcommand{\psuc}{P_{\text{suc}}}
\newcommand{\perr}{P_{\text{err}}}

\newcommand{\psia}{\psi^{(1)}}
\newcommand{\psib}{\psi^{(2)}}
\newcommand{\lambdaa}{\lambda^{(1)}}
\newcommand{\lambdab}{\lambda^{(2)}}
\newcommand{\ga}{g^{(1)}}
\newcommand{\gb}{g^{(2)}}
\newcommand{\bmu}{\bm{\mu}}
\newcommand{\bmua}{\bm{\mu}^{(1)}}
\newcommand{\bmub}{\bm{\mu}^{(2)}}
\newcommand{\mua}{\mu^{(1)}}
\newcommand{\mub}{\mu^{(2)}}
\newcommand{\mone}{\bm{1}}
\newcommand{\mzero}{\bm{0}}

\IEEEoverridecommandlockouts
\usepackage{authblk}
\title{Belief Propagation with Quantum Messages for Symmetric Q-ary Pure-State Channels}
\author[1,3]{Avijit Mandal}
\author[1,2,3]{Henry D. Pfister}
\affil[1]{Department of Electrical and Computer Engineering, Duke University}
\affil[2]{Department of Mathematics, Duke University}
\affil[3]{Duke Quantum Center, Duke University}
\begin{document}
\pagestyle{empty}
% \onecolumn
\maketitle
\begin{abstract}
Belief propagation with quantum messages (BPQM) provides a low-complexity alternative to collective measurements for communication over classical--quantum channels.
Prior BPQM constructions and density-evolution (DE) analyses have focused on binary alphabets.
Here, we generalize BPQM to symmetric q-ary pure-state channels (PSCs) whose output Gram matrix is circulant.
For this class, we show that bit-node and check-node combining can be tracked efficiently via closed-form recursions on the Gram-matrix eigenvalues, independent of the particular physical realization of the output states.
These recursions yield explicit BPQM unitaries and analytic bounds on the fidelities of the combined channels in terms of the input-channel fidelities.
This provides a DE framework for symmetric q-ary PSCs that allows one to estimate BPQM decoding thresholds for LDPC codes and to construct polar codes on these channels.

% \HP{We should pick between `q-ary symmetric'' and ``symmetric q-ary'' and be consistent. I like the latter.}\AM{fixed} 
% We show that check-node and bit-node combining preserve this symmetry up to a \emph{heralded mixture} of symmetric channels, and that the evolution can be tracked efficiently using closed-form update rules for the Gram-matrix eigenvalues (eigen lists), independent of the particular physical realization of the output states.
% Belief propagation with quantum messages (BPQM) provides a computationally efficient alternative to collective measurements for communication over classical-quantum (CQ) channels. While previous work has established the notion of BPQM for symmetric binary-input channels, extending this framework beyond binary alphabets remained a significant challenge. In this paper, we generalize BPQM to symmetric q-ary pure-state channels (PSCs), characterized by circulant Gram matrices. We demonstrate that the evolution of these channels under bit-node and check-node combining can be tracked efficiently using the eigenvalues of the Gram matrix, independent of the specific output states. We derive explicit eigenvalue update rules for the recursive channel combining and establish bounds on the average channel fidelity. Furthermore, we leverage these results to develop a density evolution algorithm for symmetric q-ary PSCs. For symmetric q-ary PSCs, we apply this framework to compute decoding thresholds of q-ary low-density parity-check (LDPC) codes and to design and decode polar codes.
\end{abstract}
\vspace{-2mm}

\section{Introduction}

Coding for classical--quantum (CQ) channels was initiated by Holevo~\cite{holevo1998capacity} and by Schumacher and Westmoreland~\cite{schumacher1997sending}, who showed that reliable communication is achievable at rates up to the capacity via collective measurements over long blocks.  For many structured CQ channels, including pure-state constellations arising in optical communication~\cite{krovi2015optimal,da2013achieving}, implementing such collective measurements is a key bottleneck in practice.  This motivates structured receiver architectures whose complexity scales favorably with blocklength while retaining strong performance guarantees. A recently proposed algorithm called decoded quantum interferometry (DQI) reduces certain structured optimization tasks to reversible decoding problems, underscoring renewed interest in efficient quantum decoding primitives~\cite{jordan2025optimization}.

A prominent family of capacity-achieving codes in the classical setting is polar codes, introduced by Ar{\i}kan~\cite{arikan2009channel}.  Wilde and Guha extended polarization to binary-input CQ channels~\cite{wilde2012polar}, which was subsequently generalized for arbitrary $q$-ary CQ channels and CQ-MACs~\cite{nasser2018polar}.  While these works establish the existence of capacity-achieving constructions, the decoding measurements for general CQ channels are not known to admit a simple, message-passing structure analogous to classical belief propagation.

Belief propagation with quantum messages (BPQM) was introduced by Renes~\cite{Renes-njp17} as a quantum analogue of classical BP for decoding binary linear codes on pure-state channels (PSCs).  Piveteau and Renes proved that BPQM is simultaneously optimal for bit- and block-error probability on tree factor graphs and reduced the decoding complexity by introducing a quantum reliability register~\cite{Piveteauarxiv21}.  Subsequent work developed BPQM-style density-evolution (DE) analyses for symmetric binary CQ channels and used them to study LDPC and polar constructions~\cite{brandsen2022belief,mandal2023belief,mandal2024polarcodescqchannels}, with further extensions to turbo-style constructions~\cite{piveteau2025efficient}.  However, extending BPQM beyond binary alphabets has remained challenging. 

In this paper, we generalize BPQM to \emph{symmetric $q$-ary pure-state channels} (PSCs) whose output Gram matrices are circulant.  Our main observation is that, for this class, the effect of check-node and bit-node combining can be tracked \emph{entirely} via the list of eigenvalues (or eigen list) of the Gram matrix, independent of the particular choice of output states within the isometry equivalence class.  Concretely, we derive explicit eigen list update rules for check and bit node combining and provide corresponding BPQM unitaries, and discuss how the same operations extend to \emph{heralded mixtures} of symmetric PSCs via controlled unitaries. We leverage the eigen list recursion to develop practical DE procedures for both polar and LDPC ensembles on symmetric $q$-ary PSCs, enabling polar-code design for a target block error rate and estimating LDPC decoding thresholds.
We characterize symmetric Holevo information, channel fidelity and pretty good measurement error rate using the Gram matrix for symmetric  $q$-ary PSCs. We also obtain \emph{upper bounds} on the fidelity of the combined channels in terms of the input-channel fidelities 
via eigen list update rules from BPQM.
% (specializing and potentially tightening general $q$-ary CQ bounds).
\vspace{-1mm}

\section{Background}
\vspace{-2mm}
\subsection{Notation}
The natural numbers are denoted by $\mathbb{N}=\left\{ 1,2,\ldots\right\} $ and  $\mathbb{N}_0=\mathbb{N}\cup \{0\}$. For $m\in \mathbb{N}_0$,  we use the shorthand $[m]\coloneqq\left\{ 0,\ldots,m-1\right\} $. We use boldface notation to denote vectors (e.g. $\by,\bw$).
We use $\mzero$ and $\mone$ to denote all zero and all one vector or matrix respectively if it is clear from the context. The identity matrix is denoted by $\mI$. For a matrix $M$, the element in the $i\textsuperscript{th}$ row and the $j\textsuperscript{th}$ column is denoted by $M_{i,j}$, while $M_{i,:}$ refers to the $i\textsuperscript{th}$ row of $M$. The canonical basis is represented using $\{\ket{j}\}_{j\in \mathbb{N}_0}$.
% \vspace{-2mm}
\subsection{Preliminaries}
% \AM{moved some description to appendix}

\begin{defn}
    A classical-quantum (CQ) channel $W\colon x\rightarrow \rho_{x}$ takes a classical input $x$ from a finite size alphabet $\cX$ and outputs a density matrix $\rho_{x}$. 
\end{defn}
For the rest of the paper, we set the input alphabet size as $q$ and denote the input alphabet as $[q]$. All the addition and multiplication operations are defined with respect to modulo $q$, i.e., we consider $\mathbb{Z}_q$.  While we restrict our analysis to prime q for simplicity and space limitations, the BPQM update rules derived in Sec.~\ref{sec: BPQM q-ary symmetric PSC} rely solely on the diagonalization of the Gram matrix via the Fourier vectors (or more specifically, Discrete Fourier Transform). These results generalize naturally to any channel whose symmetry corresponds to a finite abelian group $\cG$ by replacing the DFT with the group character table.
% $\cG\cong \mathbb{Z}_{q_1^a}\times \mathbb{Z}_{q_2^b}\dots$ (where $q_1,q_2$ are primes) 
%\HP{Assume prime and add remark.}\AM{what about now}\HP{Great!}
\begin{defn}
    Two CQ channels $W$ and $W'$ with the same input alphabet size are isometrically (or unitarily) equivalent if there exists an isometry (or unitary) $V$, such that $W'(j)=VW(j)V^{\dagger}$  $\forall j\in [q]$.
\end{defn}

% \begin{defn}
%     Two CQ channels $W$ and $W'$ with the same input alphabet size are unitarily equivalent if there exists a unitary $U$, such that $W'(j)=UW(j)U^{\dagger}$ $\forall j\in [q]$.
% \end{defn}

\begin{defn}
    A CQ channel $W\colon j\rightarrow \rho_{j}$ is called a pure state channel (PSC), if the output state $\rho_{j}=\ketbra{\psi_j}{\psi_j}$ is a pure state i.e. a rank-1 matrix $\forall j\in [q]$.
\end{defn}
% When $W$ is a PSC, we denote the output as $\ket{\psi_{x}}$ to emphasize single quantum state as output $\forall x\in \cX$. 

 \vspace{-2mm}
\subsection{Symmetric $q$-ary Pure State Channel}

% \AM{we are using $G$ for both Gram matrix and the polar matrix $G_2$, how to deal with that?} \HP{Use $K$ for polar code matrix (referencing kernel).}
Consider a collection of quantum states $
\{\ket{\psi_u}\}_{u\in [q]}$. Let the Gram matrix $G$ be defined by $G_{i,j}=\bra{\psi_{i}}\ket{\psi_j}.$
% \begin{align}
%     G_{i,j}=\bra{\psi_{i}}\ket{\psi_j}.
% \end{align}
Denote the first row of $G$ as $ G_{0,:}=[g_0,\dots ,g_{q-1}].$
When the Gram matrix is circulant, it satisfies $ G_{i,j}=g_{(j-i)}$ where $(j-i)$ is defined in modulo $q$.
% \begin{align*}
%     G_{i,j}=g_{(j-i)}.
% \end{align*}
% where we define $(k) \coloneqq k \bmod q$. \AM{since we already mention, all operations are in modulo q, probably we can avoid this?}

\begin{defn}
    A PSC $W \colon u\rightarrow \ketbra{\psi_u}{\psi_u}$ is called symmetric $q$-ary PSC, if the Gram matrix formed by the states $\{\ket{\psi_u}\}_{u\in [q]}$ is circulant. 
\end{defn}
Define $\omega\coloneqq e^{2\pi\mathrm{i}/q}$ to be a primitive $q$\textsuperscript{th} root of unity. The orthogonality relation that we are going to use throughout the paper is that for some fixed $i,j\in [q]$, $ \sum_{k\in [q]}\omega^{(i-j)k}=q\delta_{i,j}.$
% \begin{align}\label{eq:omega orthogonality}
%     \sum_{k\in [q]}\omega^{(i-j)k}=q\delta_{i,j}.
% \end{align}
We denote the $m$\textsuperscript{th}  Fourier vector $\ket{v_m}$ componentwise with $\braket{j}{v_m}=\frac{1}{\sqrt{q}}\omega^{jm}$ for all $j,m\in [q]$. Using the orthogonality relation of $\omega$, it follows that $ \braket{v_m}{v_{m'}} =\delta_{m,m'}.$
% $\ket{v_{m}}\coloneqq \frac{1}{\sqrt{q}}\begin{bmatrix}
%         1\\
%         \omega^{m}\\
%         \dots\\
%         \omega^{(q-2)m}\\
%         \omega^{(q-1)m}
%     \end{bmatrix}$ $\forall m\in [q]$ .
% % \begin{align}\label{eq:fourier-vector}
% %     \ket{v_{m}}\coloneqq \frac{1}{\sqrt{q}}\begin{bmatrix}
% %         1\\
% %         \omega^{m}\\
% %         \dots\\
% %         \omega^{(q-2)m}\\
% %         \omega^{(q-1)m}
% %     \end{bmatrix}
% % \end{align}
\begin{lem}\label{lem:gram eigenvector}
    For $m\in[q]$, $\ket{v_{m}}$ is an eigenvector of $G$ satisfying $G \ket{v_m} = \lambda_m \ket{v_m}$ where $\lambda_{m}=\sum_{j\in [q]}g_{j}\omega^{jm}$ $\forall m\in [q]$.
\end{lem}
The proof can be found in Appendix~\ref{Appendix:Gram matrix properties}. 
% in the detailed version of this paper \cite{Mandal-arxiv26}.
We denote the eigenvalues of the Gram matrix $G$ using  the ordered list $\blambda=[\lambda_0,\dots,\lambda_{q-1}]$ which we refer to as the eigen list. Note that we use the specific ordering implied by the stated Fourier transform and its correspondence to eigenvectors. Using the eigen list $\blambda$ and the Fourier vectors $\{\ket{v_{m}}\}_{m\in [q]}$ define the state $\ket{\psi_u}$ as
\begin{align}\label{eq:psi-fourier-form}
    \ket{\psi_u} &\coloneqq \frac{1}{\sqrt{q}}\sum_{j\in [q]}\sqrt{\lambda_{j}}\omega^{-uj}\ket{v_{j}}
\end{align}
\begin{lem}\label{lem: eigenlist trace relation}
    For Gram matrix $G$ with eigen list $\blambda=[\lambda_0,\dots,\lambda_{q-1}]$, it holds $\sum_{u\in [q]}\lambda_u  =q$.
\end{lem}
\begin{proof}
    Observe that 
    \begin{align*}
        \sum_{u\in [q]}\lambda_u = \Tr(G)
        = \sum_{u\in [q]}\braket{\psi_u}{\psi_u}
         = q \qquad  \qedhere
    \end{align*}
\end{proof}
% \HP{the ordered list}
% \HP{Spectrum usually indicates an unordered set. The term list is often used to indicate order. Add a remark to note that we use a specific Fourier ordering and correspondence to eigenvectors.}
% \HP{Define what you mean by $q$-ary symmetric pure-state channels before lemma. In typical lingo, the channel is group symmetric for the cyclic group.  You results probably extend to any abelian group though.}
% \AM{I am referring symmetry as Gram matrix generated using the quantum states being circulant . should we set $q$ to be prime? then state while our results hold for cyclic group $\mathbb{Z}_q$, it can be easily extended to $G\cong \mathbb{Z}_{q_1^a}\times \mathbb{Z}_{q_2^b}\dots$ (i.e. any abelian group)? }
% \HP{Yes, we should assume prime in this paper but add a remark that almost all results extend to any abelian group.} 

% \HP{We should be consistent between eigenlist and eigen list.  I prefer the former.}\AM{yes, I changed those to eigen list }
The immediate consequence of Lemma~\ref{lem: eigenlist trace relation} is that $\bmu=[\mu_0,\dots,\mu_{q-1}]$ where $\mu_j=\frac{\lambda_j}{q}$ $\forall j\in [q]$ forms a probability distribution on $[q]$ which we refer to as normalized eigen list.
\begin{lem}\label{lem:canonical state representation}
    % All unitarily equivalent q-ary symmetric pure-state channels $W\colon u\rightarrow \ketbra{\psi_{u}}{\psi_u}$ with $u\in [q]$ can be characterized using eigen list $\blambda$.
    All symmetric $q$-ary pure-state channels $W\colon u\rightarrow \ketbra{\psi_{u}}{\psi_u}$ with $u\in [q]$ are determined (up to isometric equivalence) by their eigen list $\blambda$.
\end{lem}
% \HP{All symmetric q-ary pure-state channels are determined (up to isometric equivalence) by their eigen list.}\AM{updated}
The proof can be found in Appendix~\ref{Appendix:Gram matrix properties}. 
% of \cite{Mandal-arxiv26}. 
For circulant Gram matrix, to describe PSC $W$,  we will use quantum states $\{\ket{\psi_u}\}_{u\in [q]}$ defined in \eqref{eq:psi-fourier-form}.
% \HP{Must use eqref with mathtools showonly refs.  I recommend not using "Eq. 1" and just using eqref which gives (1).}
\subsection{Information Measures}
For symmetric $q$-ary PSC $W$, for the rest of the paper we consider input distribution for $W$ to be uniform i.e. each classical input is chosen with probability $\frac{1}{q}$.
To characterize $W$ with uniform input ensemble, we consider two information measures, namely symmetric Holevo information which we denote as $I(W)$ and channel fidelity which is denoted by $F(W)$. More detailed description of $I(W)$ and $F(W)$ for  channel $W$ is provided in Appendix~\ref{Appendix: Information Measures}.
% of \cite{Mandal-arxiv26}. 
Below lemma establishes the relation between $I(W)$ and $F(W)$ with Gram matrix $G$.
\begin{lem}\label{lem:Holevo information and channel fidelity}
Consider the symmetric  $q$-ary pure-state channel $W\colon u\rightarrow \ketbra{\psi_{u}}{\psi_u}$ with $u\in [q]$ with circulant Gram matrix $G$ and eigenvalue spectrum $\blambda=[\lambda_0,\dots,\lambda_{q-1}]$. Then, for the uniform distribution over input alphabets $u\in [q]$, the symmetric Holevo information $I(W)$ and the channel fidelity $F(W)$ , for channel $W$ satisfies 
\begin{align*}
    I(W) &\ = H(\bmu)\\
  % F(W)  &\  =\frac{1}{q(q-1)}\left(\sum_{u\in [q]}\lambda_{u}^2-q\right)
  F(W) & =\frac{1}{q-1}\sum_{u=1}^{q-1}|g_{u}|
\end{align*}
where $\bmu=[\mu_0,\dots,\mu_{q-1}]$ is the normalized eigen list with $\mu_j=\frac{\lambda_j}{q}$ $\forall j\in [q]$ and $H(\bmu)$ is the Shannon entropy for $\bmu$.
\end{lem}
The proof of this lemma is provided in Appendix~\ref{Appendix: Information Measures}. 
% in \cite{Mandal-arxiv26}.

\subsection{Pretty Good Measurement on Symmetric $q$-ary PSC}\label{sec:pgm}
% \begin{defn} 
% An $m$-outcome \emph{projective measurement} of a quantum system in $\mathcal{H}_q$ is defined by a set of $m$ orthogonal projection matrices $\Pi_{j} \in \mathbb{C}^{q\times q}$ satisfying $\Pi_{i}\Pi_{j}=\delta_{i,j}\Pi_{i}$ and $\sum_{j}\Pi_{j}=\mI$.
% We denote such a measurement by $\hat{\Pi}=\{\Pi_{j}\}|_{j=1}^{m}$.
% \end{defn}
A pretty good measurement (PGM) \cite{hausladen1994pretty,holevo1978asymptotically}, which is also known as the square root measurement, is defined as follows, consider the $q$-ary CQ channel $W\colon j\rightarrow \rho_j$ with uniform input distribution on $[q]$. For the output state $\rho_j$ for $j\in [q]$, the measurement operator is $M_{j}=\frac{1}{q}\bar{\rho}^{-\frac{1}{2}}\rho_j\bar{\rho}^{-\frac{1}{2}}$ where $\bar{\rho}=\sum_{j\in [q]}\frac{1}{q}\rho_j$ and the inverse is taken on the support of $\bar{\rho}$.  This $q$-outcome measurement $\{M_j\}_{j\in [q]}$, is shown to be optimal \cite{eldar2002quantum} in terms of  minimizing the probability of state discrimination error for geometrically uniform quantum states. In this paper, to analyze the performance of BPQM,
we will use PGM to obtain the optimal symbol error probability. The following lemma establishes the relation between the error probability for discriminating the output states of symmetric $q$-ary PSC $W$ in terms of eigen list of the Gram matrix.  
\begin{lem}\label{lem:pgm}
Consider symmetric $q$-ary pure-state channel $W\colon u\rightarrow \ketbra{\psi_{u}}{\psi_u}$ with $u\in [q]$ with circulant Gram matrix $G$ and eigen list $\blambda=[\lambda_0,\dots,\lambda_{q-1}]$. Then, for the uniform distribution over input alphabets $u\in [q]$, the probability of error $\perr(W)$ for distinguishing the states optimally using pretty good measurement satisfies 
\begin{align*}
    \perr(W)=1-\left(\frac{1}{q}\sum_{u\in [q]}\sqrt{\lambda_u}\right)^{2}
\end{align*}
\end{lem}

The proof of this lemma is provided in Appendix~\ref{Appendix: pgm}.
% of \cite{Mandal-arxiv26}.

%\HP{You should add citations or at least notes but who you plan to cite.  I would Eldar-Forney here though maybe there are also earlier refs.}
%\AM{added a few citations on this}
\section{BPQM on Symmetric PSC}\label{sec: BPQM q-ary symmetric PSC}
% Check node and bit node combining operations are the primary building blocks to design decoders for linear codes whose factor graph is a tree. More specifically, in the context of estimating a code symbol, the associated tree factor graph can be decomposed into a sequence of check node and bit node channel combining operations.  For CQ channels, these channel combining operations are defined below.
Check and bit node combining operations are the fundamental building blocks for decoding linear codes on tree factor graphs. Specifically, symbol estimation decomposes the graph into a sequence of these operations. Their definitions for CQ channels follow
\begin{defn}[Check node combining]
   For symmetric $q$-ary PSCs $W_1$ and $W_2$ the check node channel combining operation $W_1\cnop W_2$, $\forall l\in [q]$ is defined as
   \begin{align*}
       [W_1\cnop W_2](l)\coloneqq\frac{1}{q}\sum_{u\in [q]}W_1(u)\otimes W_2(u-l).
   \end{align*}

\end{defn}

\begin{defn}[Bit node combining]
   For symmetric $q$-ary PSCs $W_1$ and $W_2$ the bit  node channel combining operation $W_1\vnop W_2$ , $\forall l\in [q]$ is defined as
   \begin{align*}
       [W_1\vnop W_2](l)\coloneqq W_1(l)\otimes W_2(l).
   \end{align*}
\end{defn}
 
For channel combining operations, BPQM finds a unitary transformation, such that the output of the combined channel becomes a CQ channel output  in the first register with  classical side information contained in the second register. For symmetric $q$-ary PSCs, BPQM converts combined channel output to a symmetric $q$-ary PSC output in the first register. The resultant channel can be characterized via an ensemble of symmetric PSCs known as heralded mixtures of symmetric $q$-ary PSCs. 

Characterizing these PSCs is often a hard task without a suitable choice of a unitary. The key property of BPQM is that, when we apply a suitable unitary on the output of combined channel of two heralded mixtures of  symmetric $q$-ary PSCs the output state can also be described as a heralded mixture of symmetric $q$-ary PSCs (see Appendix~\ref{appendix:hearlded PSC}). 
% of \cite{Mandal-arxiv26}). 
Thus, we apply BPQM repeatedly based on the channel combining operations associated with the factor graph, and apply a qudit measurement at the end (PGM in the context of minimizing symbol error probability) to estimate the root symbol without sacrificing optimality.

 If we can characterize the BPQM output of check and bit node combining operations for symmetric $q$-ary PSCs i.e. heralded mixtures of symmetric $q$-ary PSCs using the eigen list of the Gram matrices, then we can efficiently classically estimate the output channel. This realization is the key feature of BPQM, which lets us estimate the optimal symbol error probability without the need of implementing the actual quantum decoder which we describe in more detail in the density evolution section (Sec.~\ref{sec:density evolution}). In the following lemmas, we characterize the output symmetric $q$-ary PSCs for BPQM on check and bit node channel combining operations.
% \AM{The eigenvalue update is wrt to check node combining matrix $\begin{bmatrix}
%     1 & 0\\
%     1 & 1
% \end{bmatrix}$, should we generalize it to $\begin{bmatrix}
%     1 & 0\\
%     \alpha & 1
% \end{bmatrix}$ with $\alpha\in [q]$? }
\begin{lem}\label{lem:check node}
    Consider  symmetric $q$-ary PSCs $W_{1}$ and $W_{2}$ with Gram matrices $G^{(1)}$ and $G^{(2)}$ which have eigen lists $\blambda_{1}=[\lambda^{(1)}_{0},\dots,\lambda_{q-1}^{(1)}]$ and $\blambda_{2}=[\lambda^{(2)}_{0},\dots,\lambda_{q-1}^{(2)}]$ respectively. Then, the check node combined channel $W_{1}\cnop W_{2}$ can be decomposed into an ensemble  of symmetric $q$-ary PSCs $\{p^{\cnop}_{m},W^{\cnop}_{m}\}_{m\in [q]}$ where  $p^{\cnop}_{m}$ is the probability of the $m$\textsuperscript{th} $q$-ary PSC $W^{\cnop}_{m}$. Let $\blambda^{\cnop}_{m}=[\lambda^{(\cnop,m)}_0,\dots ,\lambda^{(\cnop,m)}_{q-1}]$ be the eigen list for $W^{\cnop}_{m}$ for $m\in [q]$. Then, for  $m\in[q]$, $p_{m}^{\cnop}$ and $\blambda^{\cnop}_{m}$ are computed using $\blambda_1$ and $\blambda_2$ as below
    \begin{align*}
        p_{m}^{\cnop} & =\frac{1}{q^2}\sum_{j\in [q]}\lambda_{(m+j)}^{(1)}\lambda_{-j}^{(2)}\\
        \lambda_{j}^{(\cnop,m)} & = \frac{1}{qp_{m}^{\cnop}}\lambda_{(m+j)}^{(1)}\lambda_{-j}^{(2)}
    \end{align*}
    \end{lem}

% \AM{Todo: change 'unitarily' argument to 'isometrically' for full generality}
    \begin{lem}\label{lem:bit node}
        Consider symmetric $q$-ary PSCs $W_{1}$ and $W_{2}$ with Gram matrices $G^{(1)}$ and $G^{(2)}$ which have eigen list $\blambda_{1}=[\lambda^{(1)}_{0},\dots,\lambda_{q-1}^{(1)}]$ and $\blambda_{2}=[\lambda^{(2)}_{0},\dots,\lambda_{q-1}^{(2)}]$ respectively. Then, the bit node combined channel $W_{1}\vnop W_{2}$ is isometrically equivalent to  symmetric $q$-ary PSC $W^{\vnop}$ with Gram matrix eigen list $\blambda^{\vnop}=[\lambda^{\vnop}_{0},\dots, \lambda^{\vnop}_{q-1}]$ such that 
        \begin{align*}
            \lambda^{\vnop}_{j}=\frac{1}{q}\sum_{k\in [q]}\lambdaa_{k}\lambdab_{j-k}
        \end{align*}
    \end{lem}
Proofs of two lemmas are provided in Appendix~\ref{Appendix:BPQM for symmetric PSC}.
% of \cite{Mandal-arxiv26}.
    \subsection{Check node and Bit node Unitary}\label{sec: check and bit node unitary description}
  We define check and bit node unitary for BPQM to be the unitaries which convert the outputs of the combined channel from symmetric $q$-ary PSCs to a heralded mixture of symmetric $q$-ary PSCs. For a heralded mixture of symmetric $q$-ary PSCs, we follow the convention that the second register  contains the classical side information. Thus, the state in the second register can be in any fixed orthonormal basis. 
    As described in the proof of Lemma~\ref{lem:check node}, the unitary $\tilde{U}^{\cnop}$ satisfies $\forall j,j'\in [q]$, 
    \begin{align}\label{eq:checknode unitary relation}
        \tilde{U}^{\cnop} \left( \ket{v_{j}}\otimes \ket{v_{j'}} \right) =\ket{v_{j+j'}}\otimes \ket{v_{-j'}}
    \end{align}
    where $\{\ket{v_j}\}_{j\in [q]}$ are the Fourier vectors and $-j'$ refers to $(-j'$ modulo $q)$. Using this we can get the matrix representation of $\tilde{U}^{\cnop}$ in $\{\ket{v_{j}}\}_{j\in [q]}$ basis. Note that in this case the first register after applying $\tilde{U}^{\cnop}$ stays in the Fourier basis $\{\ket{v_j}\}_{j\in [q]}$ which is sufficient to characterize the check node combined channel. Let $U^{\cnop}$ be the check node unitary such that the second register remains in canonical basis $\{\ket{j}\}_{j\in [q]}$. 
    Let $F$ be the DFT matrix such that $\ket{v_{j}}=F\ket{j}$ $\forall j\in [q]$ and SWAP be the unitary which swaps the states between two registers. Then $U^{\cnop}$ is implemented using $\tilde{U}^{\cnop}$ as follows
    \begin{align*}
        U^{\cnop}=(\mI\otimes F^{\dagger})(\text{SWAP}) \tilde{U}^{\cnop}
    \end{align*}
    Notice that the check node unitary $U^{\cnop}$, does not depend on the channels. The bit node unitary  $U^{\vnop}_{\blambda_1,\blambda_2}$, satisfies the following relation $\forall u\in [q]$,
    \begin{align}\label{eq:bitnode unitary relation}
U^{\vnop}_{\blambda_1,\blambda_2} \left( \ket{\psi^{(1)}_{u}}\otimes \ket{\psi^{(2)}_u} \right) =\ket{\psi^{\vnop}_u}\otimes \ket{0}
    \end{align}
    % Note that we can choose any fixed vector on the 2nd register for this operation.
    % , for example $\ket{v_0}=\frac{1}{\sqrt{q}}\mone$.
    % \HP{Use bold one transpose?} \AM{fixed} \HP{Good.  I took $\mone$ out of fraction.} 
    % $\ket{v_0}=\frac{1}{\sqrt{q}}\begin{bmatrix}
    %     1\\
    %     \dots\\
    %     1
    % \end{bmatrix}$.
    We can describe $U^{\vnop}_{\blambda_1,\blambda_2}$ as a combination of the following unitary operators. 
    Consider the unitary $U_{+}$ such that $\forall j,j'\in [q]$, it satisfies 
    \begin{align*}
        U_{+} \left(\ket{v_{j}}\otimes \ket{v_{j'}}\right)=\ket{v_{j+j'}}\otimes \ket{v_{j'}}
    \end{align*}
    Then, we get
\begin{align*}
   &  U_{+} \left( \ket{\psia_u}\otimes \ket{\psib_{u}} \right)\\
   & =  \frac{1}{q}\sum_{j\in [q]}\sum_{j'\in [q]}\sqrt{\lambdaa_{j}\lambdab_{j'}}\omega^{-uj-uj'}\ket{v_{j+j'}}\otimes\ket{v_{j'}} \\
   & = \frac{1}{q}\sum_{k\in [q]}\sum_{j'\in [q]}\sqrt{\lambdaa_{k-j'}\lambdab_{j'}}\omega^{-uk}\ket{v_{k}}\otimes\ket{v_{j'}}\\
   & = \frac{1}{q}\sum_{k\in [q]}\omega^{-uk}\ket{v_{k}} \otimes \ket{\smash{\tilde{\zeta}_{k}}}
\end{align*}
where $\ket{\smash{\tilde{\zeta}_{k}}}$ is an unnormalized vector satisfying
\begin{align*}
    \ket{\tilde{\zeta}_{k}} =\sum_{j\in [q]}\sqrt{\lambdaa_{k-j}\lambdab_{j}}\ket{v_{j}}
\end{align*}
The norm $\left\|\ket{\smash{\tilde{\zeta}_{k}}}\right\|^{2}$ satisfies 
\begin{align*}
   \left\|\ket{\smash{\tilde{\zeta}_{k}}}\right\|^{2}  = \sum_{j\in [q]}\lambdaa_{k-j}\lambdab_{j}
 = q\lambda^{\vnop}_{k}
\end{align*}
Defining $\ket{\zeta_{k}}=\frac{\ket{\tilde{\zeta}_k}}{\|\ket{\smash{\tilde{\zeta}_{k}}}\|}$, we get 
\begin{align*}
     &  U_{+} \left(\ket{\psia_u}\otimes \ket{\psib_{u}}\right) = \frac{1}{\sqrt{q}}\sum_{k\in [q]}\sqrt{\lambda^{\vnop}_{k}}\omega^{-uk}\ket{v_{k}} \otimes \ket{\zeta_{k}}
\end{align*}
Next, consider the following unitary $U^{\text{control}}_{\blambda_1,\blambda_2}$ as below
\begin{align*}
    U^{\text{control}}_{\blambda_1,\blambda_2}=\sum_{k\in [q]}\ketbra{v_k}{v_k}\otimes U_{\blambda_1,\blambda_2}^{k}
\end{align*}
where $U_{\blambda_1,\blambda_2}^{k}$ satisfies the following relation $\forall k\in [q]$
\begin{align*}
    U_{\blambda_1,\blambda_2}^{k}\ket{\zeta_{k}}= \ket{0}
\end{align*}
The unitary $U_{\blambda_1,\blambda_2}^{k}$ is constructed in the following way.
\begin{align*}
    U_{\blambda_1,\blambda_2}^{k}= \begin{cases}
        \mI \quad & ,\text{if $\ket{\zeta_{k}}=\ket{0}$}\\
        \mI-2\ketbra{\zeta'_k}{\zeta'_k} &, \text{otherwise}
    \end{cases}
\end{align*}

% When $\ket{\zeta_{k}}=\ket{0}$, we set
% \begin{align*}
%     U_{\blambda_1,\blambda_2}^{k}=\mI
% \end{align*}
% Otherwise, we obtain $U_{\blambda_1,\blambda_2}^{k}$ as
% \begin{align*}
%     U_{\blambda_1,\blambda_2}^{k}=\mI-2\ketbra{\zeta'_k}{\zeta'_k}
% \end{align*}
where $\ket{\zeta'_k}=\frac{\ket{\zeta_k}-\ket{0}}{||\ket{\zeta_k}-\ket{0}||}$. Thus setting $U^{\vnop}_{\blambda_1,\blambda_2}= U^{\text{control}}_{\blambda_1,\blambda_2}U_{+}$ satisfies the required relation in \eqref{eq:bitnode unitary relation}.
\begin{rem}
    While the check  and the bit node unitary (\eqref{eq:checknode unitary relation}, \eqref{eq:bitnode unitary relation}) are defined for the canonical choice of quantum states defined in \eqref{eq:psi-fourier-form}, we can easily extend the construction of unitaries for all isometrically (or unitarily) equivalent symmetric $q$-ary PSCs with same Gram matrix. Let $W\colon j\rightarrow \ket{\psi_j}$ be a symmetric $q$-ary PSC whose outputs are canonical quantum states for Gram matrix eigen list $\blambda$ as shown in \eqref{eq:psi-fourier-form}. Let $W'\colon j\rightarrow \ket{\psi_{j}'}$ be another  symmetric $q$-ary PSC such that $\ket{\psi_{j}'}=V_{W'}\ket{\psi_j}$ $\forall j\in [q]$ and some isometry (or unitary) $V_{W'}$. Then the check node and bit node unitary satisfy the following relations
    \begin{align*}
        U^{\cnop}_{W'} & = (V_{W'}\otimes \mI)U^{\cnop}(V_{W'}^{\dagger}\otimes V_{W'}^{\dagger})\\
        U^{\vnop}_{\blambda_1,\blambda_2,W'} & = (V_{W'}\otimes \mI)U^{\vnop}_{\blambda_1,\blambda_2}(V_{W'}^{\dagger}\otimes V_{W'}^{\dagger}).
    \end{align*}
    % Note that for channel $W'$, we only have to apply $U^{\cnop}(W')$ or $U^{\vnop}_{\blambda_1,\blambda_2}(W')$ at the first round of check node or bit node combining.
\end{rem}
\begin{rem}
   While the check and the bit node unitary $U^{\cnop}$ and $U^{\vnop}_{\blambda_1,\blambda_2}$ are described for single round of BPQM on symmetric $q$-ary PSCs, it can be easily extended for heralded mixtures of PSCs by designing controlled unitaries which apply BPQM operation on symmetric $q$-ary PSCs associated with the two heralded mixtures of symmetric $q$-ary PSCs in a controlled manner. More generally, they can be described via the following forms up to SWAP  symmetry
   \begin{align*}
      & U^{\cnop}_{\text{control}} =U^{\cnop}\otimes \mI\\
      & U^{\vnop}_{\text{control},\{\blambda_{x_1}\}_{x_1\in \cX_1},\{\blambda_{x_2}\}_{x_2\in\cX_2}} \\
      & = \sum_{x_1\in \cX_1,x_2\in \cX_2}U^{\vnop}_{\blambda_{x_1},\blambda_{x_2}}\otimes \ketbra{x_1}{x_1}\otimes \ketbra{x_2}{x_2}.
   \end{align*}
   where $\cX_1$ and $\cX_2$ are finite sized alphabets.
\end{rem}

\begin{figure*}[t]
% \vspace{-2mm}

    % \HP{Shorten legend titles on left and increase font size on both.}
    % \AM{fixed }
    \centering
    % --- First Figure ---
    \begin{minipage}{0.48\linewidth}
        \centering
        \includegraphics[width=0.8\linewidth]{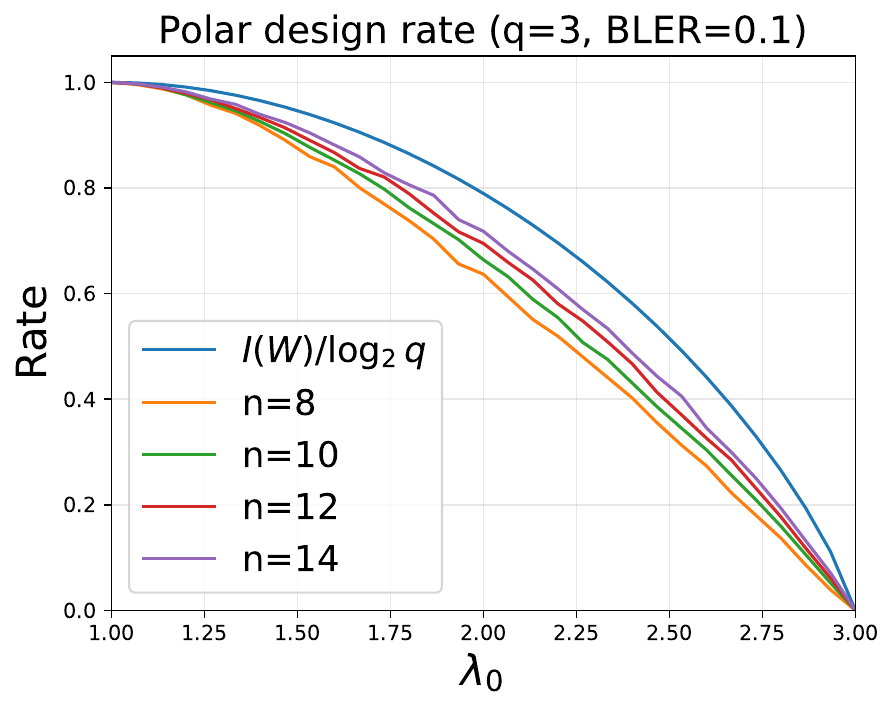}
        \vspace{-4mm} % Adjust vertical space if needed
        \caption{Plot comparing rate $\frac{|\cA|}{N}$ with $N=2^n$ of polar codes for $q=3$}
        \label{plot:polar rate vs capacity q=3}
    \end{minipage}
    \hfill % Adds flexible space between the two minipages
    % --- Second Figure ---
    \begin{minipage}{0.48\linewidth}
        \centering
        \includegraphics[width=0.9\linewidth]{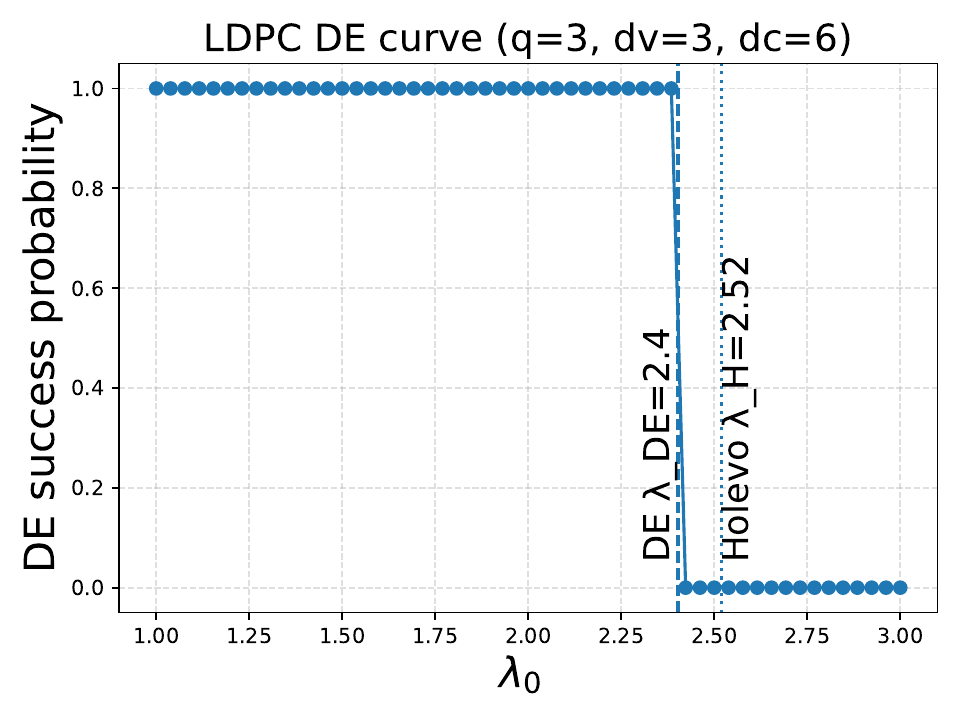}
        \vspace{-4mm} 
        \caption{Threshold plot for $(3,6)$ LDPC code with $q=3$}
        \label{plot: ldpc threshold curve}
    \end{minipage}
    \vspace{-2mm}
\end{figure*}

\vspace{-2mm}
\subsection{Fidelity Bounds}
% \AM{whom to cite here}\HP{I would cite this tutorial \url{https://dl.acm.org/doi/book/10.5555/1203431} but the main link is down now.}
Bhattacharyya coefficients play a key role for analyzing the asymptotic behavior of linear codes with tree factor graphs on classical channels \cite{sason2006performance}. Since channel fidelity generalizes the notion of Bhattacharyya coefficients to CQ channels, obtaining fidelity bounds help us analyze these codes on CQ channels  \cite{wilde2012polar,nasser2018polar}. In \cite{nasser2018polar}, authors obtained the fidelity bounds for check and bit node combined channels to analyze polar codes on arbitrary $q$-ary CQ channels. Below lemma obtains the fidelity bounds for bit and check node combined channels based on BPQM eigen list update rules for symmetric $q$-ary PSCs.
    \begin{lem}\label{lem:fidelity check and bit node bounds}
          Consider symmetric $q$-ary PSCs $W_{1}$ and $W_{2}$ with Gram matrices $G^{(1)}$ and $G^{(2)}$ which have eigen list $\blambda_{1}=[\lambda^{(1)}_{0},\dots,\lambda_{q-1}^{(1)}]$ and $\blambda_{2}=[\lambda^{(2)}_{0},\dots,\lambda_{q-1}^{(2)}]$ respectively. Then, for the check node and bit node combined channels $W_{1}\cnop W_{2}$ and $W_{1}\vnop W_{2}$, the following relations hold
          \begin{align*}
          F(W_{1}\vnop W_{2}) & \leq  (q-1)F(W_1)F(W_{2})\\
               F(W_{1}\cnop W_{2}) & \leq F(W_1)+F(W_{2})+ (q-1)F(W_1)F(W_{2}).
              \end{align*}
    \end{lem}
 The proof of the above lemma can be found in Appendix~\ref{Appendix: fidelity bounds}. 
 % of \cite{Mandal-arxiv26}. 
 These fidelity bounds obtained  improve the fidelity upper bounds obtained in \cite{nasser2018polar} for symmetric $q$-ary PSCs. In Appendix~\ref{appendix:hearlded PSC},
  % of \cite{Mandal-arxiv26}, 
 we show that the same fidelity bounds also hold for the check and bit node channel combining of two arbitrary heralded mixtures of symmetric $q$-ary PSCs.

 \vspace{-2mm}

\section{Density Evolution}\label{sec:density evolution}
Density evolution (DE) is a standard asymptotic analysis technique for belief-propagation decoding of tree factor graph based codes, where one tracks how the distribution of messages evolves with iterations~\cite{richardson2008modern}. In the context of decoding on CQ channels, a Monte Carlo based DE was introduced in~\cite{brandsen2022belief} which proposed using a generalization of BPQM called Paired Measurement BPQM (PM-BPQM) on binary symmetric CQ channels.   
In this paper, we follow the same principle of DE where message update is based on the eigen list of the Gram matrices for symmetric $q$-ary PSCs.\vspace{-2mm}
\subsection{Polar Codes}

Binary Polar codes are constructed by using the polar transform  kernel $K_N$ \cite{arikan2009channel} to encode a vector $u \in \{0,1\}^N$ whose values are freely chosen on a subset $\mathcal{A} \subseteq [N]$ of information positions but restricted to have fixed values on the complementary set $\mathcal{A}^c = [N] \setminus A$ of frozen positions.
For the frozen positions, the fixed values are shared with the receiver in advance to aid the decoding process. In the $q$-ary case, we follow the same philosophy. We still consider the polar transform of length $N=2^n$, but the matrix $K_2$ is defined in $\mathbb{Z}_q$. The successive cancellation (SC) decoder of polar code realizes a CQ channel corresponding to the $i\textsuperscript{th}$ input $u_i$ based on the classical side information $u_{1}^{i-1}$. More generally the effective channel $W^{i}$ for input $u_{i}\in [q]$ realized  by the SC decoder is written as follows 
\begin{align*}
    W^{(i)}_{N}(u_{i})=\frac{1}{q^{N-1}} \!\!\!\! \sum_{u_{\sim i}\in [q]^{N-1}} \!\!\!\! \ketbra{u^{i-1}_{1}}{u^{i-1}_{1}}\otimes \left( \bigotimes_{j=1}^N W([\bm{u} K_N]_j)\! \right)\!.
\end{align*}
 We can analyze these channels essentially using check node and bit node channel combining operations. More specifically, the $W^{i}_{N}$ channel satisfies the following recursions 
\begin{align*}
     W_{N}^{(2i-1)} &\ =W^{(i)}_{N/2}\cnop W^{(i)}_{N/2}\\
     W_{N}^{(2i)} &\ =W^{(i)}_{N/2}\vnop W^{(i)}_{N/2}.
\end{align*}
For a code with length $N=2^n$, it requires $n$ levels of channel recursions based on check and bit node channel combining operations. A detailed description of polar coding on  symmetric $q$-ary PSC is provided in Appendix~\ref{Appendix:polar coding on symmetric PSC}.
% of \cite{Mandal-arxiv26}.
\begin{lem}
    Consider the polarization channel $\{W^{(i)}_{N}\}_{i=1}^{N}$ after $n$ rounds of polarization and target block error rate $\epsilon$, then the information set $\cA$ is designed as the largest index set such that $4\sum_{i\in \cA}\perr(W^{(i)}_N)\leq \epsilon$.
\end{lem}
The multiplicative factor $4$ comes from the noncommuting nature of quantum states  in the context of sequential measurements \cite{gao2015quantum}. 
% \AM{did not mention Monte Carlo much it seems} \HP{Cite arXiv version and include in appendix there.}
Based on Monte Carlo DE, we obtain the design rate  $\frac{|\cA|}{N}$  of the polar code for a target block error rate $\epsilon$. In this case, we choose the eigen list for channel $W$ to be $\blambda=[\lambda_0,\frac{q-\lambda_0}{q-1},\dots,\frac{q-\lambda_0}{q-1}]$ and $\lambda_0$ is the maximum eigenvalue of $G$. $\lambda_0=q$ implies the channel is useless as $I(W)=0$ and $\lambda_0=1$ implies the channel is perfect. In Fig.~\ref{plot:polar rate vs capacity q=3}, we vary $\lambda_0$ from 1 to $q$ with $q=3$ and obtain the design rate for different block lengths  via BPQM based on DE (as described in Algorithm~\ref{algo:polar qudit de})
% of \cite{Mandal-arxiv26}) 
with a target block error rate $\epsilon=0.1$. We also plot the curve of symmetric Holevo information $I(W)$. As expected, we observe that as we increase the block length, the design rate curve nears $I(W)$. More results based on DE are provided in Appendix~\ref{Appendix:DE results}.
% of \cite{Mandal-arxiv26}.
% In plot~\ref{plot: polar design curve}, we the average channel error rate obtained from the DE for each channel $W^{i}$. To compare this with respect to the different block lengths, we consider the normalized channel index $\frac{i}{N}$ of the channel $W^{i}$. Then, we sort the channel according to the ascending order of channel error rate to get the normalized channel rank and obtain the curve for each block length. We observe that, as increase the block length, we get sharp threshold behavior in terms channel error rate which emphasizes the channel polarization behavior for q-ary symmetric PSCs. 

%\HP{Add Lemma or Theorem that characterizes the error rate of polar codes (e.g., using non-commutative union bound).}

% \begin{figure}
%     \centering
%     \includegraphics[width=0.9\linewidth]{pdf_figs/rate_vs_capacity.pdf}
%     \vspace{-4mm}
%     \caption{Plot comparing rate $\frac{|\cA|}{N}$ of polar codes with different block lengths for $q=3$}
%     \label{plot:polar rate vs capacity q=3}
% \end{figure}
\vspace{-1.5mm}
\subsection{LDPC Codes}
% Low-density parity-check (LDPC) codes and iterative decoding were introduced by Gallager in 1960~\cite{gallager1962low} but they did not attract widespread interest until the introduction of Turbo codes~\cite{berrou1993near} and the rediscovery of LDPC codes~\cite{mackay2002good} four decades later.
% Belief propagation (BP), in its general form, was introduced by Pearl in 1982~\cite{Pearlaaai82} as an efficient algorithm to exactly compute marginals for tree-structured probability models. BP works by passing messages between neighboring nodes and it was subsequently shown that BP includes both Gallager's iterative decoding and turbo decoding as special cases~\cite{mceliece1998turbo,Kschischangjsac98}.

% DE was first extended to the context of decoding LDPC codes on binary symmetric CQ channels in \cite{brandsen2022belief}, which utilized the parameter update based on PM-BPQM.
 Based on the eigen list update for symmetric $q$-ary PSC from Lemma~\ref{lem:check node},\ref{lem:bit node}, we obtain the thresholds in terms of eigen value of Gram matrix of channel $W$ via DE. In Algorithm~\ref{alg:ldpc_de_qudit} of Appendix~\ref{sec: DE Algorithms}, 
 % of \cite{Mandal-arxiv26}, 
 the description of DE algorithm for $(dv,dc)$ LDPC codes on symmetric $q$-ary PSC with arbitrary eigen list $\blambda$ has been provided where $dv$ corresponds to bit degree and $dc$ corresponds to check degree.  We consider the channel $W$ with eigen list $\blambda=[\lambda_0,\frac{q-\lambda_0}{q-1},\dots,\frac{q-\lambda_0}{q-1}]$ to characterize the threshold in terms of a single parameter $\lambda_0$. In Fig.~\ref{plot: ldpc threshold curve}, we obtain the $\lambda_0$ threshold for $q=3$ and $(3,6)$ regular LDPC codes. Notice that the symmetric Holevo information bound in terms of $\lambda_0$ for rate $\frac{1}{2}$ code is $2.52$, while BPQM threshold is $2.4$. More threshold results for different $(dv,dc)$ pairs are shown in Appendix~\ref{Appendix:DE results}. 
 % of \cite{Mandal-arxiv26}.
% \begin{figure}
%     \centering
%     \includegraphics[width=0.9\linewidth]{pdf_figs/ldpc_thresh.pdf}
%     \caption{Threshold plot for $(3,6)$ LDPC code with $q=3$ }
%     \label{plot: ldpc threshold curve}
% \end{figure}
\vspace{-2mm}
\section{Conclusion}
In this work, we propose a novel approach to study BPQM on symmetric $q$-ary PSCs based on the eigenvalues of the Gram matrix. The resulting check and bit node eigen list recursions yield explicit BPQM unitaries and analytic fidelity bounds, enabling practical density evolution without large-scale quantum state simulation. We demonstrate the approach by (i) designing q-ary polar codes for a target block error rate and (ii) estimating BPQM decoding thresholds for regular LDPC ensembles. Extending these tools from prime q to general finite-abelian symmetries via character-based diagonalization is a natural next step.
% Based on DE we obtained the first efficient construction of polar codes and obtained the upper bound of thresholds of LDPC codes on these channels. While we provide analysis and results for $q$ being a prime, in subsequent versions, we aim to provide detailed description for arbitrary abelian group symmetry. 
\clearpage
% \HP{Bibliography is a bit short.  Are there not other papers worth mentioning?  Background on applications?  Recent work by Guha on other approaches?}\AM{added some}
\IEEEtriggeratref{13}
\printbibliography

@article{nasser2018polar,
  title={Polar codes for arbitrary classical-quantum channels and arbitrary cq-macs},
  author={Nasser, Rajai and Renes, Joseph M},
  journal={IEEE Transactions on Information Theory},
  volume={64},
  number={11},
  pages={7424--7442},
  year={2018}
}

@article{wilde2012polar,
  title={Polar codes for classical-quantum channels},
  author={Wilde, Mark M and Guha, Saikat},
  journal={IEEE Transactions on Information Theory},
  volume={59},
  number={2},
  pages={1175--1187},
  year={2012},
  publisher={IEEE}
}

@article{Renes-njp17,
  title={Belief propagation decoding of quantum channels by passing quantum messages},
  author={Renes, Joseph M},
  journal={New Journal of Physics},
  volume={19},
  number={7},
  pages={072001},
  year={2017},
  publisher={IOP Publishing},
  url = {http://arxiv.org/abs/1607.04833}
}

@inproceedings{brandsen2022belief,
  title={Belief propagation with quantum messages for symmetric classical-quantum channels},
  author={Brandsen, S and Mandal, Avijit and Pfister, Henry D},
  booktitle={2022 IEEE Information Theory Workshop (ITW)},
  pages={494--499},
  year={2022},
  organization={IEEE}
}

@article{arikan2009channel,
  title={Channel polarization: A method for constructing capacity-achieving codes for symmetric binary-input memoryless channels},
  author={Arikan, Erdal},
  journal={IEEE Transactions on information Theory},
  volume={55},
  number={7},
  pages={3051--3073},
  year={2009},
  publisher={IEEE}
}

@article{gao2015quantum,
  title={Quantum union bounds for sequential projective measurements},
  author={Gao, Jingliang},
  journal={Physical Review A},
  volume={92},
  number={5},
  pages={052331},
  year={2015},
  publisher={APS}
}

@article{holevo1998capacity,
  title={The capacity of the quantum channel with general signal states},
  author={Holevo, Alexander S},
  journal={IEEE Transactions on Information Theory},
  volume={44},
  number={1},
  pages={269--273},
  year={1998},
  publisher={IEEE}
}

@book{richardson2008modern,
  title={Modern coding theory},
  author={Richardson, Tom and Urbanke, Ruediger},
  year={2008},
  publisher={Cambridge university press}
}

@article{schumacher1997sending,
  title={Sending classical information via noisy quantum channels},
  author={Schumacher, Benjamin and Westmoreland, Michael D},
  journal={Physical Review A},
  volume={56},
  number={1},
  pages={131},
  year={1997},
  publisher={APS}
}

@article{krovi2015optimal,
  title={Optimal measurements for symmetric quantum states with applications to optical communication},
  author={Krovi, Hari and Guha, Saikat and Dutton, Zachary and da Silva, Marcus P},
  journal={Physical Review A},
  volume={92},
  number={6},
  pages={062333},
  year={2015},
  publisher={APS}
}

@article{da2013achieving,
  title={Achieving minimum-error discrimination of an arbitrary set of laser-light pulses},
  author={da Silva, Marcus P and Guha, Saikat and Dutton, Zachary},
  journal={Physical Review A},
  volume={87},
  number={5},
  pages={052320},
  year={2013},
  publisher={APS}
}

@article{gallager1962low,
  title={Low-density parity-check codes},
  author={Gallager, Robert},
  journal={IRE Transactions on information theory},
  volume={8},
  number={1},
  pages={21--28},
  year={1962},
  publisher={IEEE}
}

@misc{mandal2024polarcodescqchannels,
      title={Polar {C}odes for {CQ} {C}hannels: Decoding via {B}elief-{P}ropagation with {Q}uantum {M}essages}, 
      author={Avijit Mandal and S. Brandsen and Henry D. Pfister},
      year={2024},
      url={https://arxiv.org/abs/2401.07167}, 
}

@inproceedings{berrou1993near,
  title={Near Shannon limit error-correcting coding and decoding: Turbo-codes. 1},
  author={Berrou, Claude and Glavieux, Alain and Thitimajshima, Punya},
  booktitle={Proceedings of ICC'93-IEEE International Conference on Communications},
  volume={2},
  pages={1064--1070},
  year={1993},
  organization={IEEE}
}

@article{mackay2002good,
  title={Good error-correcting codes based on very sparse matrices},
  author={MacKay, David JC},
  journal={IEEE transactions on Information Theory},
  volume={45},
  number={2},
  pages={399--431},
  year={2002},
  publisher={IEEE}
}

@article{mceliece1998turbo,
  title={Turbo decoding as an instance of Pearl's" belief propagation" algorithm},
  author={McEliece, Robert J. and MacKay, David J. C. and Cheng, Jung-Fu},
  journal={IEEE Journal on selected areas in communications},
  volume={16},
  number={2},
  pages={140--152},
  year={1998},
  publisher={IEEE}
}

@inproceedings{Pearlaaai82,
  title={Reverend {B}ayes on Inference Engines: A Distributed Hierarchical Approach},
  author={Pearl, Judea},
  year={1982}
}

@article{Kschischangjsac98,
  title={Iterative decoding of compound codes by probability propagation in graphical models},
  author={Kschischang, Frank R. and Frey, Brendan J.},
  volume={16},
  number={2},
  pages={219--230},
  year={1998},
  publisher={IEEE}
}

@article{Piveteauarxiv21,
  title={Quantum message-passing algorithm for optimal and efficient decoding},
  author={Piveteau, Christophe and Renes, Joseph M},
  journal={arXiv preprint arXiv:2109.08170},
  year={2021}
}

@article{hausladen1994pretty,
  title={A ‘pretty good’measurement for distinguishing quantum states},
  author={Hausladen, Paul and Wootters, William K},
  journal={Journal of Modern Optics},
  volume={41},
  number={12},
  pages={2385--2390},
  year={1994},
  publisher={Taylor \& Francis}
}

@article{holevo1978asymptotically,
  title={On asymptotically optimal hypotheses testing in quantum statistics},
  author={Holevo, Alexander Semenovich},
  journal={Teoriya Veroyatnostei i ee Primeneniya},
  volume={23},
  number={2},
  pages={429--432},
  year={1978},
  publisher={Russian Academy of Sciences, Steklov Mathematical Institute of Russian~…}
}

@inproceedings{csacsouglu2010entropy,
  title={An entropy inequality for q-ary random variables and its application to channel polarization},
  author={{\c{S}}a{\c{s}}o{\u{g}}lu, Eren},
  booktitle={2010 IEEE International Symposium on Information Theory},
  pages={1360--1363},
  year={2010},
  organization={IEEE}
}

@article{piveteau2025efficient,
  title={Efficient and optimal quantum state discrimination via quantum belief propagation},
  author={Piveteau, Christophe and Renes, Joseph M},
  journal={arXiv preprint arXiv:2509.19441},
  year={2025}
}

@inproceedings{mandal2023belief,
  title={Belief-propagation with quantum messages for polar codes on classical-quantum channels},
  author={Mandal, Avijit and Brandsen, Sarah and Pfister, Henry D},
  booktitle={2023 IEEE International Symposium on Information Theory (ISIT)},
  pages={613--618},
  year={2023},
  organization={IEEE}
}

@article{jordan2025optimization,
  title={Optimization by decoded quantum interferometry},
  author={Jordan, Stephen P and Shutty, Noah and Wootters, Mary and Zalcman, Adam and Schmidhuber, Alexander and King, Robbie and Isakov, Sergei V and Khattar, Tanuj and Babbush, Ryan},
  journal={Nature},
  volume={646},
  number={8086},
  pages={831--836},
  year={2025},
  publisher={Nature Publishing Group UK London}
}

@article{sason2006performance,
  title={Performance analysis of linear codes under maximum-likelihood decoding: A tutorial},
  author={Sason, Igal and Shamai, Shlomo},
  year={2006},
  journal={Now Publishers Inc}
}

@article{eldar2002quantum,
  title={On quantum detection and the square-root measurement},
  author={Eldar, Yonina C and Forney, G David},
  journal={IEEE Transactions on Information Theory},
  volume={47},
  number={3},
  pages={858--872},
  year={2002},
  publisher={IEEE}
}

\clearpage
\onecolumn
\begin{appendices}
\section{Quantum Preliminaries}
A \emph{qudit} represents a quantum system with $q$ perfectly distinguishable states. A $q$ dimensional Hilbert space is denoted by $\mathbb{C}^{q}$.  Such a vector $\ket{\psi}\in \hn$ with unit norm is called a \emph{pure state}. Let $\mathbb{C}^{q\times q}$ denote the vector space of $q \times q$ complex matrices. % (operators on $\hn$) 
 Let $\hop^q \subset \mathbb{C}^{q\times q}$ denote the subset of operators (i.e., matrices) mapping $\hn$ to $\mathbb{C}^q$ that are Hermitian. Let $\psd^q\subset \hop^q$ and $\pd^q \subset \psd^q$ denote the subset of Hermitian operators mapping $\hn$ to $\hn$ which are positive semidefinite and positive definite, respectively.
Let $\dop^q \subset \psd^q$ denote the set of operators on $\hn$ which are positive semidefinite with unit trace. In quantum, a \emph{density matrix} $\rho\in \dop^q$ can be defined by an ensemble of pure states $\Psi = \{p_i,\ket{\psi_i}\}_{i\in [q]}$ where $p_{i}$ is the probability of choosing the pure state $\ket{\psi_{i}}$ and
\[\rho=\sum_{i \in [q]}p_{i}\ketbra{\psi_{i}}{\psi_{i}}.\]
The unitary evolution of a quantum pure state $\ket{\psi} \in \hn$ is described by the mapping $\ket{\psi} \mapsto U \ket{\psi}$, where $U\in \mathbb{C}^{q\times q}$ is a unitary. For the pure state ensemble $\Psi$, this evolution results in the modified ensemble $\Psi'=\{p_{i},U\ket{\psi_{i}}\}_{i\in[q]}$ whose density matrix is \vspace{-1.5mm}
\begin{align*}
    \rho' = \sum_{i\in [q]} p_{i}U\ketbra{\psi_{i}}{\psi_{i}}U^{\dagger}=U\rho U^{\dagger},
\end{align*}
where $U^\dagger$ is the Hermitian transpose of $U$.
    \section{Gram Matrix Properties}\label{Appendix:Gram matrix properties}
    \subsection{Proof of Lemma~\ref{lem:gram eigenvector}}
    \begin{proof}
  Consider the $i$\textsuperscript{th} row $G_{i}$ of the Gram matrix $G$. Then $G_{i,:}\ket{v_{m}}$ satisfies
  \begin{align*}
      G_{i,:}\ket{v_{m}}&\ = \frac{1}{\sqrt{q}}\sum_{j\in [q]}g_{(j-i)}\omega^{jm}\\
      &\ = \frac{1}{\sqrt{q}}\sum_{j'\in [q]}g_{j'}\omega^{(j'+i)m}\\
      &\ =\frac{\omega^{im}}{\sqrt{q}}\sum_{j'\in [q]}g_{j'}\omega^{j'm}\\
      &\ =\lambda_{m}\frac{\omega^{im}}{\sqrt{q}}
  \end{align*}
  Thus, we get 
  \begin{align*}
      G\ket{v_{m}} &\ =\frac{\lambda_{m}}{\sqrt{q}}\begin{bmatrix}
        1\\
        \omega^{m}\\
        \dots\\
        \omega^{(q-2)m}\\
        \omega^{(q-1)m}  
      \end{bmatrix}\\
      &\ = \lambda_{m}\ket{v_{m}} \qedhere
  \end{align*}
\end{proof}

\subsection{Eigen list trace relation}
\begin{lem}\label{lem: eigenlist trace  square relation}
    For gram matrix $G$ with eigen list $\blambda=[\lambda_0,\dots,\lambda_{q-1}]$, the following relations hold 
    \begin{align*}
        \sum_{u,u'\in[q]}|\braket{\psi_u}{\psi_{u'}}|^{2}  =\sum_{u\in [q]}\lambda_{u}^{2}
    \end{align*}
\end{lem}
% \subsection{Proof of Lemma~\ref{lem: eigenlist trace relation}}
\begin{proof}
 
Similarly, from matrix $G^2$ we get
\begin{align*}
    \sum_{u\in [q]}\lambda_{u}^2 & =\Tr(G^2)\\
    & =\Tr(G^{\dagger}G)\\
    & = \sum_{u,u'\in [q]}|G_{u,u'}|^{2}\\
    & =\sum_{u,u'\in [q]}\left|\braket{\psi_{u}}{\psi_{u'}}\right|^{2} \qedhere
\end{align*}
\end{proof}
% \subsection{Proof of Lemma~\ref{lem: eigenlist trace relation}}
% \begin{proof}
%        Observe that 
%        \begin{align*}
%         \sum_{u\in [q]}\lambda_u & = \Tr(G)\\
%         &= \sum_{u\in [q]}\braket{\psi_u}{\psi_u}\\
%         & = q
%         \end{align*}
  
% \end{proof}
\subsection{Proof of Lemma~\ref{lem:canonical state representation}}
\begin{proof}
Recall the canonical choice of state $\ket{\psi_u}$ for Gram matrix $G$ with eigen list $\blambda$ as
\begin{align}
    \ket{\psi_u} &\coloneqq \frac{1}{\sqrt{q}}\sum_{j\in [q]}\sqrt{\lambda_{j}}\omega^{-uj}\ket{v_{j}}
\end{align}
The inner product $\braket{\psi_{u'}}{\psi_{u}}$ satisfies 
\begin{align*}
    \braket{\psi_{u'}}{\psi_{u}} &\ =\frac{1}{q}\sum_{j\in [q]}\sum_{j'\in [q]}\sqrt{\lambda_{j}\lambda_{j'}}\omega^{(u'j'-uj)}\braket{v_{j}}{v_{j'}}\\
    &\ = \frac{1}{q}\sum_{j\in [q]}\lambda_{j}\omega^{j(u'-u)}\\
    &\ = \frac{1}{q}\sum_{j\in [q]}\sum_{k\in [q]}g_{k}\omega^{kj}\omega^{j(u'-u)}\\
    &\ =\frac{1}{q}\sum_{k\in [q]}g_{k}\sum_{j\in [q]}\omega^{j(k+u'-u)}\\
    &\ =\frac{1}{q}\sum_{j\in [q]}g_{u-u'}\\
    &\ =g_{u-u'}= G_{u',u}
\end{align*}
Since two classical quantum PSCs $W\colon u\rightarrow \ket{\psi_{u}}$ and $W'\colon u\rightarrow \ket{\psi_{u}'}$ with $u\in [q]$ are  isometrically (or unitarily) equivalent when there exists an isometry (or unitary) $V$ such that 
\begin{align*}
    \ket{\psi_{u}'}=V\ket{\psi_u}, \forall u\in[q],
\end{align*}
all the isometrically (or unitarily) equivalent symmetric $q$-ary PSCs with Gram matrix are well defined using eigenvalue spectrum $\blambda$ associated with Gram matrix $G$. 
\end{proof}

\subsection{Relation between Gram matrix elements and the eigen list}
\begin{lem}\label{lem:gram element eigenvalue relation}
    For a gram matrix $G$ with eigen list $\blambda = [\lambda_0,\dots,\lambda_{q-1}]$, the elements in the first row of $G$ are
    \begin{align*}
        g_{i}=\frac{1}{q}\sum_{j\in [q]}\lambda_{j}\omega^{-ij}
    \end{align*}
\end{lem}
\begin{proof}
    Using $\lambda_j=\sum_{k}g_{k}\omega^{kj}$, we get 
    \begin{align*}
        \frac{1}{q}\sum_{j\in [q]}\lambda_{j}\omega^{-ij} & = \frac{1}{q}\sum_{j\in [q]}\left(\sum_{k\in [q]}g_{k}\omega^{kj}\right)\omega^{-ij}\\
        & = \frac{1}{q}\sum_{k\in [q]}g_{k}\sum_{j\in [q]}\omega^{(k-i)j}\\
        & = \frac{1}{q}\sum_{k\in [q]}g_{k}q\delta_{k,i}\\
        & = g_{i} \qedhere
    \end{align*}
\end{proof}
\section{Information Measures}\label{Appendix: Information Measures}
The symmetric Holevo information measures the maximum amount of classical information that can be sent through the channel $W$ when input distribution is uniform.  More formally, for a CQ channel $W$, the symmetric Holevo information is computed as 
\begin{align*}
    I(W)=S\Bigg(\frac{1}{q}\sum_{u\in [q]}W(u)\Bigg)-\sum_{u\in [q]}\frac{1}{q}S\Bigg(W(u)\Bigg)
\end{align*}
where $S(\rho)=-\Tr(\rho\log\rho)$ is von Neumann entropy for density matrix $\rho$.
When $W\colon u\rightarrow\ketbra{\psi_u}{\psi_u}$ $\forall u\in [q]$ is a PSC, the symmetric Holevo information satisfies
\begin{align*}
    I(W) & =S\Bigg(\frac{1}{q}\sum_{u\in [q]}\ketbra{\psi_u}{\psi_u}\Bigg)-\frac{1}{q}\sum_{u\in [q]}S(\ketbra{\psi_u}{\psi_u})\\
   &  = S\Bigg(\frac{1}{q}\sum_{u\in [q]}\ketbra{\psi_u}{\psi_u}\Bigg)
\end{align*}
where we use the fact that $S(\ketbra{\psi_u}{\psi_u})=0$ $\forall u\in [q]$.
The channel fidelity characterizes expected overlap between the non identical output quantum states of channel $W$. Formally, it is written as 
\begin{align*}
    F(W)
    & =\frac{1}{q(q-1)}\!\sum_{u,u'\in [q]:, u\neq u'}\!\!\!\!\!\!\!\Tr\left(\sqrt{\sqrt{W(u)}W(u')\sqrt{W(u)}}\right).
\end{align*}
For symmetric PSC $W$, the expression reduces to
\begin{align*}
     F(W)=\frac{1}{q(q-1)}\sum_{u,u'\in[q]:u\neq u'}\left|\braket{\psi_u}{\psi_{u'}}\right|
\end{align*}

\subsection{Proof of Lemma~\ref{lem:Holevo information and channel fidelity}}
\begin{proof}
    % Consider the quantum states $\{\ket{\psi_{u}}\}_{u\in [q]}$ as defined in Eq~\ref{eq:psi-fourier-form}. Let $p_{u}$ corresponds to the probability of input alphabet $u$. Then the Holevo information for channel $W$ is computed as 
    % \begin{align*}
    %     I(W)&\ =S(\sum_{u\in [q]}p_{u}\ketbra{\psi_u}{\psi_u})-\sum_{u\in [q]}p_{u}S(\ketbra{\psi_u}{\psi_u})\\
    %     &\ = S(\sum_{u\in [q]}p_{u}\ketbra{\psi_u}{\psi_u})
    % \end{align*}
Consider the matrix $\Psi$ whose columns correspond to the quantum states $\{\ket{\psi_{u}}\}_{u\in [q]}$ as below
\begin{align*}
    \Psi\coloneqq \begin{bmatrix}
        \ket{\psi_0}\dots \ket{\psi_{q-1}}
    \end{bmatrix}
\end{align*}
Then we have
\begin{align*}
    \Psi\Psi^{\dagger}= \sum_{u\in [q]}\ketbra{\psi_u}{\psi_u}
\end{align*}
 The matrix $\Psi^{\dagger}\Psi=G$ as it satisfies 
\begin{align*}
    (\Psi^{\dagger}\Psi)_{i,j}=\braket{\psi_i}{\psi_j} =G_{i,j}
\end{align*}
Since the eigenvalues of $\Psi^{\dagger}\Psi$ are same as of $\Psi\Psi^{\dagger}$, we have 
\begin{align*}
    I(W) &\ = S\left(\frac{1}{q}\Psi\Psi^{\dagger}\right)\\
    &\ = -\sum_{u\in [q]}\frac{\lambda_{u}}{q}\log \frac{\lambda_u}{q}\\
    & = -\sum_{u\in [q]}\mu_u\log\mu_u\\
    & = H(\bmu)
    % &\ = \frac{\log q}{q}\sum_{u\in [q]}\lambda_u-\frac{1}{q}\sum_{u\in q}\lambda_u\log\lambda_u\\
    % &\ \stackrel{(a)}{=} \log q- \frac{1}{q}\sum_{u\in q}\lambda_u\log\lambda_u
\end{align*}
% where $(a)$ follows from the fact that 
% \begin{align*}
%     \sum_{u\in [q]}\lambda_u=\Tr(G)=q
% \end{align*}

% Next, the following relation holds
% \begin{align*}
%     \sum_{u,u'\in[q]}|\braket{\psi_u}{\psi_u'}|^{2} &\ =\Tr(G^{\dagger}G)\\
%     & = \Tr(G^{2})\\
%     & =\sum_{u\in [q]}\lambda_{u}^{2}
% \end{align*}
% \HP{It's best to start with the definition of fidelity in this type of derivation especially because there are two competing definitions (squared or unsquared).}

% \HP{You should compare also the the formulas for non-binary bhattacharyya parameter used (e.g. see appendix of https://webee.technion.ac.il/people/idotal/papers/journal/highEntropyPolarMemory.pdf) to analzye polar codes.}
% The average Fidelity $F(W)$, for the channel $W$ satisfies 
% \begin{align*}
%     F(W)=\frac{1}{q(q-1)}\sum_{u,u'\in[q]:u\neq u'}|\braket{\psi_u}{\psi_u'}|
% \end{align*}
Writing  channel fidelity $F(W)$ in terms Gram matrix elements we get 
\begin{align*}
 F(W) & =\frac{1}{q(q-1)}\sum_{u,u'\in[q]:u\neq u'}|\braket{\psi_u}{\psi_u'}|\\
 &\ =   \frac{1}{q(q-1)}\Bigg(\sum_{u,u'\in [q]:u\neq u'}|g_{u'-u}|\Bigg)\\
 &\ =\frac{1}{q(q-1)}q\bigg(\sum_{u=1}^{q-1}|g_{u}|\bigg)\\
 & =\frac{1}{q-1}\sum_{u=1}^{q-1}|g_u| \qedhere
\end{align*}
\end{proof}
\subsection{Relation Between Symmetric Holevo Information and the Channel Fidelity}
\begin{lem}\label{lem:Holevo and fidelity relation}
Consider  symmetric $q$-ary pure-state channel $W$ with circulant Gram matrix $G$ and eigenvalue spectrum $\blambda=[\lambda_0,\dots,\lambda_{q-1}]$. Then, $I(W)$ and $F(W)$ satisfy 
\begin{align*}
    % I(W) & \geq \log q-(q-1)^{2}F(W)^{2}\\
    % F(W) & \geq \frac{\sqrt{\log q-I(W)}}{q-1}\\
    F(W) & \geq \frac{\sqrt{e^{\log q-I(W)}-1}}{q-1}\\
    F(W) & \leq \sqrt{\frac{2q}{q-1}(\log q- I(W))}
\end{align*}
\end{lem}

\begin{proof}
    Observe that 
    \begin{align*}
        \sum_{u\in [q]}|g_u|^{2}=\frac{1}{q}\sum_{j\in [q]}\lambda_{j}^{2}
    \end{align*}
    Define $\bmu=[\mu_{0},\dots,\mu_{q-1}]$ where $\mu_{j}=\frac{\lambda_{j}}{q}$. Since $\lambda_{j}\geq 0$ $\forall j\in [q]$ and $\sum_{j\in [q]}\lambda_{j}=q$, $\bmu$ forms a probability distribution. Also, since $g_0=1$, we get 
    \begin{align*}
        \sum_{u=1}^{q-1}|g_u|^{2}=q\sum_{j\in [q]}\mu_{j}^{2}-1
    \end{align*}
    Consider the uniform distribution $\frac{\mone}{q}=[\frac{1}{q},\dots,\frac{1}{q}]$. Then we can compute the $\chi^{2}$ divergence as 
    \begin{align*}
        \chi^{2}\left(\bmu\,\middle\|\,\frac{\mone}{q}\right) & =\sum_{j\in [q]}\frac{(\mu_{j}-\frac{1}{q})^{2}}{\frac{1}{q}}\\
        & = q\sum_{j\in [q]}\mu_{j}^{2}-1\\
        & = \sum_{u=1}^{q-1}|g_u|^{2}
    \end{align*}
  Using the norm inequality $\|\bx\|_1 \le \sqrt{d}\|\bx\|_2$ with dimension $d=q-1$, we obtain
    % \HP{Be more specific abour CS here. LHS = dot product of (q-1) copies of g vector => no square on LHS?}
\begin{align*}
    \Bigg(\sum_{j\in 1}^{q-1}|g_{j}|\Bigg)^{2} & \leq (q-1)\sum_{j=1}^{q-1}|g_{j}|^{2}\\
    \implies (q-1)^{2}F(W)^{2} & \leq (q-1)\chi^{2}\left(\bmu \,\middle\|\, \frac{\mone}{q}\right) 
\end{align*}

% \HP{Probably better to give uniform distribution a name and use that. $u$ is natural but already used.}\AM{yeah, probably we can keep $\frac{\mone}{q}$ for now, change later}
Similarly it also satisfies 
\begin{align*}
    \Bigg(\sum_{j\in 1}^{q-1}|g_{j}|\Bigg)^{2} & \geq \sum_{j=1}^{q-1}|g_{j}|^{2}\\
    \implies (q-1)^{2}F(W)^{2}  & \geq \chi^{2}\left(\bmu \, \middle \| \, \frac{\mone}{q}\right)
\end{align*}
The KL divergence $D(\bmu||\frac{\mone}{q})$ is computed as 
\begin{align*}
    D(\bmu||\frac{\mone}{q}) & =\sum_{j\in [q]}\mu_{j}\log(\frac{\mu_{j}}{\frac{1}{q}})\\
    & = \log q+\sum_{j\in [q]}\mu_{j}\log\mu_{j}\\
    % & = \log q+\sum_{j\in [q]}\frac{\lambda_{j}}{q}\log(\frac{\lambda_j}{q})\\
    & = \log q-I(W)
\end{align*}
 $D(\bmu||\frac{\mone}{q})$ satisfies 
 % \HP{FYI, the upper bound can be improved to $\ln (1+ \chi^2)$}
 % \AM{Thanks, this makes bound tighter}
% \begin{align*}
%   \frac{1}{2}||\bmu-\frac{\mone}{q}||_{1}^{2}\leq  D(\bmu||\frac{\mone}{q})\leq \chi^{2}(\bmu||\frac{\mone}{q})
% \end{align*}
\begin{align*}
  \frac{1}{2}||\bmu-\frac{\mone}{q}||_{1}^{2}\leq  D \left(\bmu \, \middle\| \, \frac{\mone}{q} \right)\leq \log(1+\chi^{2}\left(\bmu \, \middle\| \, \frac{\mone}{q}\right))
\end{align*}
Since $||\bmu-\frac{\mone}{q}||_{1}^{2}\geq ||\bmu-\frac{\mone}{q}||_{2}^{2}=\frac{1}{q}\sum_{j=1}^{q-1}|g_{j}|^{2}$, we finally get 
% \begin{align*}
%    \frac{\chi^{2}(\bmu||\frac{\mone}{q})}{2q} \leq  D(\bmu||\frac{\mone}{q})\leq \chi^{2}(\bmu||\frac{\mone}{q})
% \end{align*}
\begin{align*}
   \frac{\chi^{2}(\bmu||\frac{\mone}{q})}{2q} \leq  D \left(\bmu \, \middle\| \, \frac{\mone}{q} \right)\leq \log\left(1+\chi^{2}\left(\bmu \, \middle\| \, \frac{\mone}{q}\right) \right)
\end{align*}
So, from the upper bound we get
\begin{align*}
    \log q-I(W) & \leq \log(1+(q-1)^{2}F(W)^{2})\\
F(W) & \geq \frac{\sqrt{e^{\log q-I(W)}-1}}{q-1}
\end{align*}
Similarly from the lower bound we get 
\begin{align*}
    \frac{(q-1)}{2q}F(W)^{2} & \leq \log q-I(W)\\
   \implies  F(W) &\leq \sqrt{\frac{2q}{q-1}(\log q-I(W))} \qedhere
\end{align*}
\end{proof}
\begin{lem}\label{lem:Holevo and fidelity boundary relations}
    For  symmetric $q$-ary pure-state channel $W$ with circulant Gram matrix $G$, when $I(W)\rightarrow 0$, we have $F(W)\rightarrow 1$ and when $I(W)\rightarrow\log q$, $F(W)\rightarrow 0$.
\end{lem}
\begin{proof}
    Since $ D(\bmu||\frac{\mone}{q})=\log q-I(W)$, when $I(W)\rightarrow 0$, we have
    \begin{align*}
        \log \left(1+\chi^{2} \left(\bmu \, \middle\| \, \frac{\mone}{q}\right)\right) & \geq \log q\\
        \implies \chi^{2} \left(\bmu \, \middle\| \,\frac{\mone}{q} \right) & \geq q-1
    \end{align*}
    Using the relation $\chi^{2}(\bmu||\frac{\mone}{q})=\sum_{u=1}^{q-1}|g_u|^{2}$ and using the fact that $|g_u|\leq 1$ $\forall u\in [q]$, we get 
    \begin{align*}
        \chi^{2}\left(\bmu \, \middle\| \, \frac{\mone}{q}\right)\leq q-1.
    \end{align*}
   Together, we have $\chi^{2}(\bmu||\frac{\mone}{q})=q-1$. Hence  $\forall u\in [q]$, we get 
   \begin{align*}
       |g_u|=1
   \end{align*}
   Thus, we get 
   \begin{align*}
       F(W)=\frac{1}{q-1}\sum_{u=1}^{q-1}|g_u|=1.
   \end{align*}
   So, we get when $I(W)\rightarrow 0$, $F(W)\rightarrow 1$.\\
   Using the upper bound 
   \begin{align*}
       F(W) &\leq \sqrt{\frac{2q}{q-1}(\log q-I(W))}
   \end{align*}
   we directly get, when $I(W)\rightarrow \log q$, $F(W)\rightarrow 0$.
\end{proof}
\section{Pretty Good Measurement}\label{Appendix: pgm}
\subsection{Proof of Lemma~\ref{lem:pgm}}
\begin{proof}
    For the states $\{\ket{\psi_u}\}_{u\in [q]}$, the pretty good measurement is constructed using measurement operators $\{\Pi_u\}_{u\in [q]}$ such that 
    \begin{align*}
        \Pi_u=\ketbra{\Gamma_u}{\Gamma_u}
    \end{align*}
    where $\ket{\Gamma_u}$ = $\sqrt{\frac{1}{q}}\bar{\rho}^{-\frac{1}{2}}\ket{\psi_u}$ and the average density matrix $\bar{\rho}=\sum_{u\in [q]}\frac{1}{q}\ketbra{\psi_u}{\psi_u}$. Thus, for uniform input distribution we have 
    \begin{align*}
    \bar{\rho}=\frac{1}{q}\Psi\Psi^{\dagger}=\sum_{u\in [q]}\frac{\lambda_u}{q}\ketbra{v_u}{v_u}
    \end{align*}
    Thus, $\ket{\Gamma_u}$ can be written as 
    \begin{align*}
        & \ket{\Gamma_u}\\
        & =\frac{1}{\sqrt{q}}(\sum_{j'\in [q]}\sqrt{\frac{q}{\lambda_{j'}}}\ketbra{v_{j'}}{v_{j'}})(\frac{1}{\sqrt{q}}\sum_{j\in [q]}\sqrt{\lambda_{j}}\omega^{-uj}\ket{v_{j}})\\
        & = \frac{1}{\sqrt{q}}\sum_{j\in [q]}\omega^{-uj}\ket{v_{j}}
    \end{align*}
Hence, we compute $\braket{\Gamma_u}{\psi_u}$ as
\begin{align*}
    \braket{\Gamma_u}{\psi_u} & =\frac{1}{q}\sum_{j\in [q]}\sum_{j'\in [q]}\sqrt{\lambda_{j'}}\omega^{uj-uj'}\braket{v_{j}}{v_{j'}}\\
    & =\frac{1}{q}\sum_{j\in [q]}\sqrt{\lambda_{j}}
\end{align*}
Finally, the average success probability of the measurement satisfies 
\begin{align*}
    \psuc(W)& =\frac{1}{q}\sum_{u\in [q]}\Tr(\Pi_u\ketbra{\psi_u}{\psi_u})\\
    & = \frac{1}{q}\sum_{u\in [q]}|\braket{\Gamma_u}{\psi_u}|^{2}\\
    & =\frac{1}{q}\sum_{u\in [q]}\left(\frac{1}{q}\sum_{j\in [q]}\sqrt{\lambda_{j}}\right)^{2}\\
    & = \left(\frac{1}{q}\sum_{j\in [q]}\sqrt{\lambda_{j}}\right)^{2} 
\end{align*}
Thus, the error probability is 
\begin{align*}
    \perr(W)=1-\psuc(W)=1-\left(\frac{1}{q}\sum_{j\in [q]}\sqrt{\lambda_{j}}\right)^{2}
\end{align*}
\end{proof}
\section{BPQM for Symmetric PSC}\label{Appendix:BPQM for symmetric PSC}
\subsection{Proof of Lemma~\ref{lem:check node}}
    \begin{proof}
        Consider the states $\{\ket{\psia_{u}}\}_{u\in[q]}$ which are outputs of channel $W_1$ such that 
        \begin{align*}
            \ket{\psia_{u}}\coloneqq \frac{1}{\sqrt{q}}\sum_{j\in [q]}\sqrt{\lambdaa_{j}}\omega^{-uj}\ket{v_{j}}.
        \end{align*}
  Similarly,  consider the states $\{\ket{\psib_{u}}\}_{u\in[q]}$ which are outputs of channel $W_2$ such that 
        \begin{align*}
            \ket{\psib_{u}}\coloneqq \frac{1}{\sqrt{q}}\sum_{j\in [q]}\sqrt{\lambdab_{j}}\omega^{-uj}\ket{v_{j}}.
        \end{align*}
    Then, the outputs of check node combined channel $W_{1}\cnop W_{2}$ can be written  $\forall l\in [q]$ as 
    \begin{align*}
       & [W_{1}\cnop W_{2}](l)\\
       & =\frac{1}{q}\sum_{u\in [q]}\ketbra{\psia_u}{\psia_u}\otimes \ketbra{\psib_{(u-l)}}{\psib_{(u-l)}}
    \end{align*}
    Expanding $\ket{\psia_u}\otimes \ket{\psib_{u-l}}$ we get 
    \begin{align*}
      &  \ket{\psia_u}\otimes \ket{\psib_{u-l}}\\
      & =\frac{1}{q}\sum_{j\in [q]}\sum_{j'\in [q]}\sqrt{\lambdaa_{j}\lambdab_{j'}}\omega^{-uj-(u-l)j'}\ket{v_{j}}\otimes\ket{v_{j'}} 
    \end{align*}
Consider the unitary $\tilde{U}^{\cnop}$ which acts on Fourier vectors as below
\begin{align*}
    \tilde{U}^{\cnop}\ket{v_j}\otimes\ket{v_{j'}}=\ket{v_{j+j'}}\otimes\ket{v_{-j'}}
\end{align*}
Then we get
\begin{align*}
   &  \tilde{U}^{\cnop}\ket{\psia_u}\otimes \ket{\psib_{u-l}}\\
   & =  \frac{1}{q}\sum_{j\in [q]}\sum_{j'\in [q]}\sqrt{\lambdaa_{j}\lambdab_{j'}}\omega^{-uj-(u-l)j'}\ket{v_{j+j'}}\otimes\ket{v_{-j'}} 
\end{align*}
Hence $\forall l\in [q]$, we can write
\begin{align*}
   & \tilde{U}^{\cnop}[W_{1}\cnop W_{2}](l)(\tilde{U}^{\cnop})^{\dagger}\\
   & = \frac{1}{q^3}\sum_{u\in [q]}\left(\sum_{j\in [q]}\sum_{j'\in [q]}\sqrt{\lambdaa_{j}\lambdab_{j'}}\omega^{-uj-(u-l)j'}\ket{v_{j+j'}}\otimes\ket{v_{-j'}}\right) \left(\sum_{k\in [q]}\sum_{k'\in [q]}\sqrt{\lambdaa_{k}\lambdab_{k'}}\omega^{uk+(u-l)k'}\bra{v_{k+k'}}\otimes\bra{v_{-k'}}\right)\\
   & = \frac{1}{q^3}\sum_{u\in [q]}\bigg{(}\sum_{j,j',k,k'\in [q]}\sqrt{\lambdaa_{j}\lambdaa_{k}\lambdab_{-j'}\lambdab_{-k'}}\omega^{-u(j-j')+u(k-k')}
  \omega^{-lj'+lk'}\ketbra{v_{j-j'}}{v_{k-k'}}\otimes \ketbra{v_{j'}}{v_{k'}}\bigg{)}\\
   & =\frac{1}{q^3}\sum_{j,j',k,k'\in [q]}\sqrt{\lambdaa_{j}\lambdaa_{k}\lambdab_{-j'}\lambdab_{-k'}}\left(\sum_{u\in [q]}\omega^{u((k-k')-(j-j'))}\right) \omega^{-lj'+lk'}\ketbra{v_{j-j'}}{v_{k-k'}}\otimes \ketbra{v_{j'}}{v_{k'}}\\
   & =\frac{1}{q^3}\sum_{j,j',k,k'\in [q]}\sqrt{\lambdaa_{j}\lambdaa_{k}\lambdab_{-j'}\lambdab_{-k'}}\left(q\delta_{j-j',k-k'}\right) \omega^{-lj'+lk'}\ketbra{v_{j-j'}}{v_{k-k'}}\otimes \ketbra{v_{j'}}{v_{k'}}\\
   & = \frac{1}{q^2}\sum_{m,j',k'\in [q]}\sqrt{\lambdaa_{m+j'}\lambdaa_{m+k'}\lambdab_{-j'}\lambdab_{-k'}}\omega^{-lj'+lk'}\ketbra{v_{m}}{v_{m}}\otimes \ketbra{v_{j'}}{v_{k'}} \\
   & =\frac{1}{q^2}\sum_{m\in [q]}\ketbra{v_{m}}{v_{m}}\otimes  \left(\sum_{j',k'\in [q]}\sqrt{\lambdaa_{m+j'}\lambdaa_{m+k'}\lambdab_{-j'}\lambdab_{-k'}}\omega^{-lj'+lk'}\ketbra{v_{j'}}{v_{k'}}\right)
\end{align*}
Hence we can compute $p^{\cnop}_{m}$ as 
\begin{align*}
    p^{\cnop}_{m} & =\Tr\left(\bra{v_{m}}\otimes \mI \left(U^{\cnop}[W_{1}\cnop W_{2}](l)(U^{\cnop})^{\dagger}\right)\ket{v_{m}}\otimes \mI\right)\\
    & = \Tr\left(\frac{1}{q^2}\sum_{j',k'\in [q]}\sqrt{\lambdaa_{m+j'}\lambdaa_{m+k'}\lambdab_{-j'}\lambdab_{-k'}}\omega^{-lj'+lk'}\ketbra{v_{j'}}{v_{k'}}\right)\\
    & =\frac{1}{q^2}\sum_{j\in [q]}\bigg{(}\sum_{j',k'\in [q]}\sqrt{\lambdaa_{m+j'}\lambdaa_{m+k'}\lambdab_{-j'}\lambdab_{-k'}}\omega^{-lj'+lk'} \bra{v_{j}}\ketbra{v_{j'}}{v_{k'}}\ket{v_{j}}\bigg{)}\\
 & = \frac{1}{q^2}\sum_{j\in [q]}\lambdaa_{m+j}\lambdab_{-j}
\end{align*}
Consider the state $\{\ket{\psi^{(\cnop,m)}_{l}}\}_{l\in [q]}$ such that 
\begin{align*}
    \ket{\psi^{(\cnop,m)}_{l}}\coloneqq\frac{1}{\sqrt{q}}\sum_{j\in [q]}\sqrt{\lambda_{j}^{(\cnop,m)}}\omega^{-lj}\ket{v_{j}}
\end{align*}
where $\lambda_{j}^{(\cnop,m)}  = \frac{1}{qp_{m}^{\cnop}}\lambda_{(m+j)}^{(1)}\lambda_{-j}^{(2)}$. Then we get 
\begin{align}
   & \ketbra{\psi^{(\cnop,m)}_l}{\psi^{(\cnop,m)}_l}\\
   & =\frac{1}{q}\sum_{j',k'\in [q]}\sqrt{\lambda_{j'}^{(\cnop,m)}\lambda_{k'}^{(\cnop,m)}}\omega^{-lj'+lk'}\ketbra{v_{j'}}{v_{k'}}\\
    & = \frac{1}{q^2p^{\cnop}_{m}}\sum_{j',k'\in [q]}\sqrt{\lambdaa_{m+j'}\lambdaa_{m+k'}\lambdab_{-j'}\lambdab_{-k'}}\omega^{-lj'+lk'}\ketbra{v_{j'}}{v_{k'}}
\end{align}
The following relation holds for the check node combined output
\begin{align*}
 &  \frac{ \bra{v_{m}}\otimes \mI \left(\tilde{U}^{\cnop}[W_{1}\cnop W_{2}](l)(\tilde{U}^{\cnop})^{\dagger}\right)\ket{v_{m}}\otimes \mI}{\Tr\left(\bra{v_{m}}\otimes \mI \left(\tilde{U}^{\cnop}[W_{1}\cnop W_{2}](l)(\tilde{U}^{\cnop})^{\dagger}\right)\ket{v_{m}}\otimes \mI\right)}\\
   & =\frac{ \bra{v_{m}}\otimes \mI \left(\tilde{U}^{\cnop}[W_{1}\cnop W_{2}](l)(\tilde{U}^{\cnop})^{\dagger}\right)\ket{v_{m}}\otimes \mI}{p^{\cnop}_{m}}\\
   & = \frac{1}{q^2p^{\cnop}_m}\sum_{j',k'\in [q]}\sqrt{\lambdaa_{m+j'}\lambdaa_{m+k'}\lambdab_{-j'}\lambdab_{-k'}}\omega^{-lj'+lk'}\ketbra{v_{j'}}{v_{k'}}\\
   & = \ketbra{\psi^{(\cnop,m)}_l}{\psi^{(\cnop,m)}_l}.
\end{align*}
    \end{proof}
\subsection{Proof of Lemma~\ref{lem:bit node}}
    \begin{proof}
      Consider the states $\{\ket{\psia_{u}}\}_{u\in[q]}$ and $\{\ket{\psib_{u}}\}_{u\in[q]}$  which are outputs of channel $W_1$ and $W_{2}$ respectively. Consider the first rows for the Gram matrices $G^{1}$ and $G^{2}$ are as below
        \begin{align*}
             G_{0,:}^{(1)} & =[\ga_0,\dots ,\ga_{q-1}]\\
             G_{0,:}^{(2)} & =[\gb_0,\dots ,\gb_{q-1}]
        \end{align*}
        Then $\forall u\in [q]$, we have 
        \begin{align*}
            \ga_{u} & =\frac{1}{q}\sum_{j\in [q]}\lambdaa_{j}\omega^{-uj}\\
            \gb_{u} & =\frac{1}{q}\sum_{j\in [q]}\lambdab_{j}\omega^{-uj}
        \end{align*}
    For bit node combined channel $W_{1}\vnop W_{2}$, $\forall u\in [q]$ the outputs are 
    \begin{align*}
        [W_{1}\vnop W_{2}](u)= \ketbra{\psia_{u}}{\psia_u}\otimes \ketbra{\psib_u}{\psib_u}
    \end{align*}
    Thus, the Gram matrix $G^{\vnop}$ for the combined channel $W_{1}\vnop W_{2}$ satisfies 
    \begin{align*}
        G^{\vnop}_{i,i'} & =\braket{\psia_i}{\psia_{i'}}\braket{\psib_i}{\psib_{i'}}\\
        & = \ga_{i'-i}\gb_{i'-i}\\
        & =\left(\frac{1}{q}\sum_{k\in [q]}\lambdaa_{k}\omega^{-(i'-i)k}\right)\left(\frac{1}{q}\sum_{k'\in [q]}\lambdab_{k'}\omega^{-(i'-i)k'}\right)\\
        & =\frac{1}{q^2}\sum_{k,k'\in [q]}\lambdaa_{k}\lambdab_{k'}\omega^{-(i'-i)(k+k')}\\
        & =\frac{1}{q^{2}}\sum_{j\in [q]} \left(\sum_{k\in [q]}\lambdaa_{k}\lambdab_{j-k}\right)\omega^{-(i'-i)j}
    \end{align*}
    Denoting $\lambda^{\vnop}_{j}=\frac{1}{q}\sum_{k\in [q]}\lambdaa_{k}\lambdab_{j-k}$, we get 
    \begin{align*}
       G^{\vnop}_{i,i'}= \frac{1}{q}\sum_{j\in [q]}\lambda^{\vnop}_{j}\omega^{-(i'-i)j} 
    \end{align*}
    Thus $G^{\vnop}$ is circulant with eigen value spectrum $\blambda^{\vnop}=[\lambda^{\vnop}_0,\dots,\lambda^{\vnop}_{q-1} ]$ which characterizes channel $W^{\vnop}$.
    \end{proof}
\section{Channel Fidelity and Symmetric Holevo Information of Check and Bit Node Combined Channels}\label{Appendix: fidelity bounds}
    \subsection{Proof of Lemma~\ref{lem:fidelity check and bit node bounds}}
       \begin{proof}
Writing $F(W_1)$ and $F(W_2)$, in terms of Gram matrix values we get 
\begin{align}
    F(W_1) & =\frac{1}{q-1}\sum_{u=1}^{q-1}|\ga_{u}|\\
    F(W_2) & =\frac{1}{q-1}\sum_{u=1}^{q-1}|\gb_{u}|
\end{align}
The fidelity for the bit node combined channel $F(W_1\vnop W_{2})$ satisfies 
\begin{align*}
    F(W_1\vnop W_2)=\frac{1}{q-1}\sum_{u=1}^{q-1}|\ga_u\gb_u|
\end{align*}
Then we get
% \HP{Be more specific about CS here. Basic CS gives RHS with square roots and sums of squares.}
\begin{align*}
    \sum_{u=1}^{q-1}|\ga_u\gb_u| & \leq \sum_{u=1}^{q-1}|\ga_{u}|\sum_{u=1}^{q-1}|\gb_{u}|\\
    \implies F(W_{1}\vnop W_{2}) & \leq (q-1)F(W_1)F(W_2)
\end{align*}
The fidelity for check node combined channel, $F(W_1\cnop W_2)$ satisfies
\begin{align*}
    F(W_1\cnop W_2) & =\sum_{m\in [q]}p_{m}^{\cnop}F(W_{m}^{\cnop})\\
    & = \sum_{m\in [q]}p_{m}^{\cnop}\left(\frac{1}{q-1}\sum_{u=1}^{q-1}|g^{(\cnop,m)}_u|\right)
\end{align*}
The Gram matrix values $g^{(\cnop,m)}_u$ can be computed as 
\begin{align*}
  &  g^{(\cnop,m)}_u \\
  & = \frac{1}{q}\sum_{j\in [q]}\lambda^{(\cnop,m)}_{j}\omega^{-uj}\\
    & =\frac{1}{q^2p_{m}^{\cnop}}\sum_{j\in [q]}\lambdaa_{m+j}\lambdab_{-j}\omega^{-uj}\\
    & = \frac{1}{q^2p_{m}^{\cnop}}\sum_{j\in [q]}\left(\sum_{k\in [q]}\ga_{k}\omega^{(m+j)k}\sum_{k'\in [q]}\gb_{k'}\omega^{-jk'}\right)\omega^{-uj}\\
    & = \frac{1}{q^2p_{m}^{\cnop}}\left(\sum_{k,k'\in [q]}\ga_{k}\gb_{k'}\omega^{mk}\right)\sum_{j\in [q]}\omega^{j(k-k'-u)}\\
    & = \frac{1}{qp_{m}^{\cnop}}\sum_{k\in [q]}\ga_{k}\gb_{k-u}\omega^{mk}
\end{align*}
Then, we get 
\begin{align*}\label{eq:checknode fidelity gram element relation}
   F(W_1\cnop W_2) & =  \sum_{m\in [q]}p_{m}^{\cnop}\left(\frac{1}{q-1}\sum_{u=1}^{q-1}|g^{(\cnop,m)}_u|\right)\\
   & = \frac{1}{q(q-1)}\sum_{m\in [q]}\sum_{u=1}^{q-1}|\sum_{k\in [q]}\ga_{k}\gb_{k-u}\omega^{mk}|
\end{align*}
Using absolute values we obtain 
% \HP{No CS here. This is just taking absolute values.}
\begin{align*}
|\sum_{k\in [q]}\ga_{k}\gb_{k-u}\omega^{mk}|\leq \sum_{k\in [q]}  |\ga_k||\gb_{k-u}|
\end{align*}
Thus, substituting this expression we get 
\begin{align*}
    F(W_1\cnop W_2) & \leq \frac{1}{q(q-1)}\sum_{m\in [q]}\sum_{u=1}^{q-1}\sum_{k\in [q]}  |\ga_k||\gb_{k-u}|\\
    & =\frac{1}{q-1}\sum_{u=1}^{q-1}\sum_{k\in [q]}  |\ga_k||\gb_{k-u}|
\end{align*}
Next, for fixed $u$, when $k=0$ we have 
\begin{align}
    |\ga_{0}|\,|\gb_{u}|=|\gb_{u}|
\end{align}
Similarly, for fixed $u$ and  $k=u$, we get 
\begin{align*}
    |\ga_{u}|\,|\gb_{0}|=|\ga_{u}| 
\end{align*}
Hence, we get the following inequality 
\begin{align*}
   & F(W_1\cnop W_2)\\
   & \leq \frac{1}{q-1}\sum_{u=1}^{q-1}\left(|\ga_{u}|+|\gb_{u}|+\sum_{k\neq 0,u}|\ga_{k}|\,|\gb_{k-u}|\right)\\
   & = F(W_1)+F(W_2) +\frac{1}{q-1}\left(\sum_{u=1}^{q-1}\sum_{k\neq 0,u}|\ga_{k}|\,|\gb_{k-u}|\right)\\
   & \leq F(W_1)+F(W_2) +\frac{1}{q-1}\left(\sum_{k=1}^{q-1}|\ga_{u}|\right)\left(\sum_{k'\in 1}^{q-1}|\gb_{k'}|\right)\\
   & =F(W_1)+F(W_2) +(q-1) F(W_1)F(W_2)
\end{align*}
    \end{proof}
\subsection{Special Case}
 \begin{lem}\label{lem:fidelity bounds special case 1}
          Consider symmetric $q$-ary PSCs $W_{1}$ and $W_{2}$ with Gram matrices $G^{(1)}$ and $G^{(2)}$ which have eigen list $\blambda_{1}=[\lambda^{(1)}_{0},\frac{q-\lambdaa_0}{q-1}\dots,\frac{q-\lambdaa_0}{q-1}]$ and $\blambda_{2}=[\lambda^{(2)}_{0},\frac{q-\lambdab_0}{q-1}\dots,\frac{q-\lambdab_0}{q-1}]$ respectively. Then, for the check node and bit node combined channels $W_{1}\cnop W_{2}$ and $W_{1}\vnop W_{2}$, the following relations hold
          \begin{align*}
          F(W_{1}\vnop W_{2}) & = F(W_1)F(W_{2})\\
               F(W_{1}\cnop W_{2}) & \leq F(W_1)+F(W_{2})-F(W_1)F(W_{2})
              \end{align*}
    \end{lem}
    % \subsection{Proof of Lemma~\ref{lem:fidelity bounds special case 1}}
       \begin{proof}
        Observe that the Gram matrix elements $|\ga_u|$ $\forall u\in [q]$ is computed as 
        \begin{align*}
            \ga_{u} & =\frac{1}{q}\sum_{j=0}^{q-1}\lambdaa_{j}\omega^{-uj}\\
            & = \frac{1}{q}(\lambdaa_0+\sum_{j-1}^{q-1}\frac{q-\lambdaa_0}{q-1}\omega^{-uj})
        \end{align*}
        If $u\neq 0$, we get 
        \begin{align*}
            \ga_u & =\frac{1}{q}(\lambdaa_0-\frac{q-\lambdaa_0}{q-1})\\
            & =\frac{\lambdaa_0-1}{q-1}
        \end{align*}
    Hence, the channel fidelity $F(W_1)$ satisfies 
    \begin{align*}
        F(W_1) & =\frac{1}{q-1}\sum_{u=1}^{q-1}|\ga_u|\\
        & = \frac{\lambdaa_0-1}{q-1}.
    \end{align*}
    Similarly, we obtain the expression for channel fidelity $F(W_2)$ as 
    \begin{align*}
        F(W_2) & =\frac{\lambdab_0-1}{q-1}
    \end{align*}
This also implies, the following holds $\forall u\neq 0$ 
\begin{align*}
    |\ga_{u}\gb_u| & =\frac{\lambdaa_0-1}{q-1}\frac{\lambdab_0-1}{q-1}\\
    & = F(W_1)F(W_2)
\end{align*}
Thus we get 
\begin{align*}
    F(W_1\vnop W_2) & =\frac{1}{q-1}\sum_{u=1}^{q-1}|\ga_u\gb_u|\\
    & = F(W_1)F(W_2)
\end{align*}
Consider the quantity $S_{m,u}=\sum_{k\in [q]}\ga_{k}\gb_{k-u}\omega^{mk}$. From above relations we get the following equality for $u\neq 0$
\begin{align*}
    S_{0,u} & =\ga_{u}+\gb_{-u}+\sum_{k\neq 0,u}\ga_{k}\gb_{k-u}\\
    & = \frac{q-\lambdaa_0}{q-1}+\frac{q-\lambdab}{q-1}+(q-2)\frac{q-\lambdaa_0}{q-1}\frac{q-\lambdab_0}{q-1}\\
    & = F(W_1)+F(W_2)+(q-2)F(W_1)F(W_2)
\end{align*}
For $m\neq 0$, observe that 
\begin{align*}
    \sum_{k\neq 0,u}\omega^{mk}= -1-\omega^{mu}.
\end{align*}
Hence, we get
\begin{align*}
   & S_{m,u} \\
   & =\frac{q-\lambdaa_0}{q-1}\omega^{mu}+\frac{q-\lambdab}{q-1}+(-1-\omega^{mu})\frac{q-\lambdaa_0}{q-1}\frac{q-\lambdab_0}{q-1}\\
   & = F(W_2)(1-F(W_1))+F(W_1)(1-F(W_2))\omega^{mu}
\end{align*}
Using \eqref{eq:checknode fidelity gram element relation}, we get 
\begin{align*}
  &  F(W_1\cnop W_2)\\
  & =\frac{1}{q(q-1)}\sum_{m\in [q]}\sum_{u=1}^{q-1}|S_{m,u}|\\
    & = \frac{1}{q(q-1)}\sum_{u=1}^{q-1}|S_{0,u}|+\frac{1}{q(q-1)}\sum_{m=1}^{q-1}\sum_{u=1}^{q-1}|S_{m,u}|\\
    & =\frac{1}{q}(F(W_1)+F(W_2)+(q-2)F(W_1)F(W_2))\\
   & \quad +\frac{1}{q}\sum_{t=1}^{q-1}|F(W_2)(1-F(W_1))+F(W_1)(1-F(W_2))\omega^{t}|
\end{align*}
Next $\forall t\neq 0$, it satisfies 
\begin{align*}
  & | F(W_2)(1-F(W_1))+F(W_1)(1-F(W_2))\omega^{t}|\\
  & \leq F(W_2)(1-F(W_1))+F(W_1)(1-F(W_2))\\
  & = F(W_1)+F(W_2)-2F(W_1)F(W_2)
\end{align*}
Thus, we get the final inequality as 
\begin{align*}
    & F(W_1\cnop W_2)\\
  &  \leq \frac{1}{q} (F(W_1)+F(W_2)+(q-2)F(W_1)F(W_2)\\
   & \quad +(q-1)(F(W_1)+F(W_2)-2F(W_1)F(W_2)) \\
   & = F(W_1)+F(W_2)-F(W_1)F(W_2)
\end{align*}
    \end{proof}
\subsection{Symmetric Holevo Information of Bit Node Combined Channel}
  \begin{lem}\label{lem:Holevo_bitnode}
          Consider symmetric $q$-ary PSCs $W_{1}$ and $W_{2}$ with Gram matrices $G^{(1)}$ and $G^{(2)}$ which have eigen list $\blambda_{1}=[\lambda^{(1)}_{0},\dots,\lambda_{q-1}^{(1)}]$ and $\blambda_{2}=[\lambda^{(2)}_{0},\dots,\lambda_{q-1}^{(2)}]$ respectively. Then, for the check node and bit node combined channels $W_{1}\cnop W_{2}$ and $W_{1}\vnop W_{2}$, the following relations hold
          \begin{align*}
          I(W_{1}\vnop W_{2}) & =  H(\bmu^{(1)}*\bmu^{(2})
          % I(W_{1}\cnop W_{2}) & = H(\bmu^{(1)})+H(\bmu^{(2)})- H(\bmu^{(1)}*\bmu^{(2})
              \end{align*}
              where $\bmua=[\mua_{0},\dots,\mua_{q-1} ]$ and $\bmub=[\mub_{0},\dots,\mub_{q-1} ]$ such that $\mua_{j}=\frac{\lambdaa_j}{q}$ and $\mub_{j}=\frac{\lambdab_j}{q}$ $\forall j\in [q]$.
    \end{lem}

    %\HP{I think this follows from the polar code martingale relationship for symmetric channels that says $I(W_1 \vnop W_2) + I(W_1 \cnop W_2) = I(W_1)+I(W_2)$ because $I(W_1 \vnop W_2) \geq  I(W_1 \cnop W_2)$.  Though it could be interesting to prove the equality directly using eigenvalues.}
    %\AM{Yes, it may be possible to prove just using eigenvalues}
    \begin{proof}
         Define $\bmu^{\vnop}=[\mu^{\vnop}_0,\dots, \mu^{\vnop}_{q-1}]$ such that $\mu^{\vnop}_{j}=\frac{\lambda^{\vnop}_j}{q}$ $\forall j\in [q]$.
  Using Lemma~\ref{lem:bit node}, we get
        \begin{align*}
            \lambda^{\vnop}_{j} & =\frac{1}{q}\sum_{k\in [q]}\lambdaa_{k}\lambdab_{j-k}\\
            \implies \mu^{\vnop}_{j} & = \sum_{k\in [q]}\mua_{k}\mub_{j-k}\\
        \implies \bmu^{\vnop}   &  = \bmu^{(1)}*\bmu^{(2)}
        \end{align*}
    Using Lemma~\ref{lem:Holevo information and channel fidelity},  the Holevo information $I(W_{1}\vnop W_{2})$ for the bit node combined channel satisfies 
        \begin{align*}
            I(W_{1}\vnop W_{2}) & =H(\bmu^{\vnop})\\
            & =-\sum_{j\in [q]}\mu^{\vnop}_{j}\log \mu^{\vnop}_{j}\\
            & = H(\bmu^{(1)}*\bmu^{(2)})
        \end{align*}
    % Define random variables $M_{1}\sim \bmua$ and $M_{2}\sim \bmub$. Define random variable $M$ such that 
    % \begin{align*}
    %     M=M_1+M_2.
    % \end{align*}
    % Thus, $M\sim \bmu^{\vnop}$.  This also satisfies the following  
    % \begin{align*}
    %     H(M) & \geq H(M|M_{1})=H(M_{2})\\
    %     H(M) & \geq H(M|M_{2})=H(M_1)
    % \end{align*}
    % Hence we get 
    % \begin{align*}
    %  &   H(M)\geq \text{Max}(H(M_1),H(M_2))\geq \frac{H(M_1)+H(M_2)}{2}\\
    %     \implies & I(W_1\vnop W_2)\geq \text{Max}(I(W_1),I(W_2))\geq \frac{I(W_1)+I(W_2)}{2}
    % \end{align*}
    \end{proof}

\section{Heralded Mixtures of Symmetric  PSC}\label{appendix:hearlded PSC}
\begin{defn}
        A channel $W_{H}$ is heralded mixtures of symmetric PSC ,if the channel output can be characterized via an ensemble of symmetric $q$-ary PSCs. More generally, the output of channel $W_H$ $\forall j \in [q]$, can be written as 
        \begin{align*}
            W_H(j)=\sum_{x\in \cX}p_{x}W_{x}(j)\otimes \ketbra{x}{x}
        \end{align*}
    where $W_{x}:j\rightarrow \ketbra{\psi^{x}_j}{\psi^{x}_j}$ is a symmetric $q$-ary PSC with circulant Gram matrix $G_{x}$ and the eigen list $\blambda_x$ $\forall x \in \cX$, $\{\ket{x}\}_{x\in \cX}$ is an orthonormal basis and $\{p_{x}\}_{x\in \cX}$ is some probability distribution in finite sized alphabet $\cX$.
    \end{defn}
    Since $\{\ket{x}\}$ forms orthonormal basis, the channel fidelity $F(W_{H})$ satisfies
    \begin{align*}
      &  F(W_{H})\\
    & =\frac{1}{q(q-1)}\sum_{j,j'\in [q]:, j\neq j'}\Tr\left(\sqrt{\sqrt{W_H(j)}W_H(j')\sqrt{W_H(j)}}\right)\\
    & = \frac{1}{q(q-1)}\sum_{j,j'\in[q]:j\neq j'}\sum_{x\in \cX}p_x\left|\braket{\psi^{x}_j}{\psi^{x}_{j'}}\right|\\
    & =\sum_{x\in \cX}p_x \left(\frac{1}{q(q-1)}\sum_{j,j'\in[q]:j\neq j'}\left|\braket{\psi^{x}_j}{\psi^{x}_{j'}}\right|\right)\\
    & = \sum_{x\in \cX}p_x F(W_x)\\
    & = \mathbb{E}_{x}(F(W_{x}))
    \end{align*}
Similarly, the symmetric Holevo information $I(W_H)$ satisfies 
\begin{align*}
    I(W_{H}) & = S\Bigg(\frac{1}{q}\sum_{j\in [q]}W_H(j)\Bigg)-\sum_{j\in [q]}\frac{1}{q}S\Bigg(W_H(j)\Bigg)\\
    & = S\Bigg(\frac{1}{q}\sum_{j\in [q]}\sum_{x\in \cX}p_{x}W_x(j)\otimes \ketbra{x}{x}\Bigg)-\sum_{j\in [q]}\frac{1}{q}S\Bigg(\sum_{x\in \cX}p_{x}W_{x}(j)\otimes \ketbra{x}{x}\Bigg)\\
    & =\sum_{x\in \cX}p_x \left(S\Bigg(\frac{1}{q}\sum_{j\in [q]}W_x(j)\Bigg)-\sum_{j\in [q]}\frac{1}{q}S\Bigg(W_x(j)\Bigg)\right)\\
    & =\sum_{x\in \cX}p_x I(W_x)\\
    & = \mathbb{E}_{x}[I(W_x)]
\end{align*}
    % \subsection{Proof of Lemma~\ref{lem fidelity bound for heralded PSC}}
\begin{lem}\label{lem fidelity bound for heralded PSC}
    Consider two heralded mixtures of symmetric PSCs $W_{H}^{(1)}$ and $W_{H}^{(2)}$ whose outputs are of the forms as below 
    \begin{align*}
        W_H^{(1)}(j) & =\sum_{x_1\in \cX_1}p_{x_1}^{(1)}W_{x_1}^{(1)}(j)\otimes \ketbra{x_1}{x_1}\\
        W_H^{(2)}(j) & =\sum_{x_2\in \cX_2}p_{x_2}^{(2)}W_{x_2}^{(2)}(j)\otimes \ketbra{x_2}{x_2}
    \end{align*}
$\forall j\in [q]$ and $\cX_1$, $\cX_2$ are finite sized alphabets, then the following holds 
\begin{align*}
    F(W_{H}^{(1)}\vnop W_{H}^{(2)}) & \leq (q-1)F(W_{H}^{(1)})F(W_{H}^{(2)})\\
    F(W_{H}^{(1)}\cnop W_{H}^{(2)}) & \leq F(W_H^{(1)})+F(W_{H}^{(2)})+(q-1)F(W_{H}^{(1)})F(W_{H}^{(2)})
\end{align*}
\end{lem}
    \begin{proof}
From the definition of bit node combining we get 
\begin{align*}
    F(W_{H}^{(1)}\vnop W_{H}^{(2)}) & =\sum_{x_1\in \cX_1,x_2\in \cX_2}p_{x_1}^{(1)}p^{(2)}_{x_2}F(W_{x_1}^{(1)}\vnop W_{x_2}^{(2)})\\
    & = \mathbb{E}_{x_1\in\cX_1,x_2\in \cX_2}[F(W_{x_1}^{(1)}\vnop W_{x_2}^{(2)})]
\end{align*}
Hence, using the upper bound of  $F(W_{x_1}^{(1)}\vnop W_{x_2}^{(2)})$ for all $x_1\in \cX_1,x_2\in \cX_2$ pairs we get 
\begin{align*}
    & F(W_{H}^{(1)}\vnop W_{H}^{(2)})\\
    & \leq (q-1)\sum_{x_1\in \cX_1,x_2\in \cX_2}p_{x_1}^{(1)}p^{(2)}_{x_2}F(W_{x_1}^{(1)})F(W_{x_2}^{(2)})\\
     & = (q-1)\left(\sum_{x_1\in \cX_1}p_{x_1}^{(1)}F(W_{x_1}^{(1)})\right)\left(\sum_{x_2\in \cX_2}p^{(2)}_{x_2}F(W_{x_2}^{(2)})\right)\\
     & = (q-1)F(W_{H}^{(1)})F(W_{H}^{(2)})
\end{align*}
Similarly for check node combining it satisfies 
\begin{align*}
    F(W_{H}^{(1)}\cnop W_{H}^{(2)}) & =\sum_{x_1\in \cX_1,x_2\in \cX_2}p_{x_1}^{(1)}p^{(2)}_{x_2}F(W_{x_1}^{(1)}\cnop W_{x_2}^{(2)})\\
    & = \mathbb{E}_{x_1\in \cX_1,x_2\in \cX_2}[F(W_{x_1}^{(1)}\cnop W_{x_2}^{(2)})]
\end{align*}
\begin{align*}
    & F(W_{H}^{(1)}\vnop W_{H}^{(2)})\\
    & \leq \sum_{x_1\in \cX_1,x_2\in \cX_2}p_{x_1}^{(1)}p^{(2)}_{x_2}\Bigg(F(W_{x_1}^{(1)})+F(W_{x_2}^{(2)})+ (q-1)F(W_{x_1}^{(1)})F(W_{x_2}^{(2)})\Bigg)\\
    & = \sum_{x_1\in \cX_1}p_{x_1}^{(1)}F(W_{x_1}^{(1)})+\sum_{x_2\in \cX_2}p^{(2)}_{x_2}F(W_{x_2}^{(2)})\\
    & \quad + (q-1)\left(\sum_{x_1\in \cX_1}p_{x_1}^{(1)}F(W_{x_1}^{(1)})\right)\left(\sum_{x_2\in \cX_2}p^{(2)}_{x_2}F(W_{x_2}^{(2)})\right)\\
    & = F(W_H^{(1)})+F(W_{H}^{(2)})+(q-1)F(W_{H}^{(1)})F(W_{H}^{(2)})
\end{align*}
\end{proof}
\begin{lem}\label{lem: Holevo information for heralded PSC}
    Consider two heralded symmetric PSCs $W_{H}^{(1)}$ and $W_{H}^{(2)}$ whose outputs are of the forms as below 
    \begin{align*}
        W_H^{(1)}(j) & =\sum_{x_1\in \cX_1}p_{x_1}^{(1)}W_{x_1}^{(1)}(j)\otimes \ketbra{x_1}{x_1}\\
        W_H^{(2)}(j) & =\sum_{x_2\in \cX_2}p_{x_2}^{(2)}W_{x_2}^{(2)}(j)\otimes \ketbra{x_2}{x_2}
    \end{align*}
$\forall j\in [q]$, then the following holds 
\begin{align*}
   I(W_{H}^{(1)}\vnop W_{H}^{(2)})=  \mathbb{E}_{x_1\in \cX_1,x_2\in \cX_2}[H(\bmu^{(1)}_{x_1}*\bmu^{(2)}_{x_2})]
\end{align*}
\end{lem}
\begin{proof}
    The proof the directly follows from Lemma~\ref{lem:Holevo_bitnode}
    \begin{align*}
        I(W_{H}^{(1)}\vnop W_{H}^{(2)}) & = \mathbb{E}_{x_1\in \cX_1,x_2\in \cX_2}[I(W^{(1)}_{x_1}\vnop W^{(2)}_{x_2})]\\
        & = \mathbb{E}_{x_1\in \cX_1,x_2\in \cX_2}[H(\bmu^{(1)}_{x_1}*\bmu^{(2)}_{x_2})]
    \end{align*}
\end{proof}
\section{Polar Coding on Symmetric PSC}\label{Appendix:polar coding on symmetric PSC}
In 2009, Arikan introduced Polar codes as the first deterministic construction
of capacity-achieving codes for binary memoryless symmetric (BMS) channels~\cite{arikan2009channel}. In the context of binary polar codes, the polar transform of length $N=2^{n}$ is obtained by $K_{N}\triangleq B_{N}K_{2}^{\otimes n}$ where $B_{N}$ is $N\times N$  bit reversal matrix~\cite[Sec.~VII.B]{arikan2009channel} and $K_{2}^{\otimes n}$ is $n$-fold tensor product of $2\times 2$ binary matrix 
\begin{align*}
    K_{2}\triangleq\begin{bmatrix}
     1 & 0\\
     1 & 1
    \end{bmatrix}.
\end{align*}
Consider the input  pair $(u_1,u_2)$ and the polar transformed pair $(x_1,x_2)$ defined by the matrix $K_2$ such that they satisfy 
\begin{align*}
    x_1 & =u_1+u_2\\
    x_2 & = u_2.
\end{align*}
Suppose, we send inputs $x_1$ and $x_2$ through the symmetric $q$-ary PSCs $W_1$ and $W_2$. Then we can construct the effective CQ channel $W^{-}$ which takes $u_1\in [q]$ as input and outputs a quantum state which can be written as 
\begin{align}\label{eq:check node polarization}
    W^{-}(u_1) & =\frac{1}{q}\sum_{u_2\in [q]}W_1(u_1+u_2)\otimes W_2(u_2)\nonumber\\
    & =\frac{1}{q}\sum_{u\in [q]}W_1(u)\otimes W_2(u-u_1)\nonumber\\
    & = [W_1\cnop W_2](u_1)
\end{align}
The CQ channel seen  by $u_2$ utilizes the fact that input $u_1$ is revealed as classical side information. Thus,  the effective channel $W^+$ which takes input $u_2\in [q]$ outputs the quantum state which can be written as 
\begin{align*}
    W^+(u_2) & =\frac{1}{q}\sum_{u_1\in [q]}\ketbra{u_1}{u_1}\otimes W_1(u_1+u_2)\otimes W_2(u_2)
\end{align*}
Since, the channels $W_1$ and $W_2$ are symmetric, there exists a collection of unitaries $\{U_{a}\}_{a\in [q]}$ such that 
\begin{align*}
    U_{a}W_{1}(u_1)U_{a}^{\dagger}=W_1(u_1+a)
\end{align*}
Define the unitary $\cU$ as below
\begin{align*}
    \cU= \sum_{u_1\in [q]}\ketbra{u_1}{u_{1}}\otimes U_{u_1}^{\dagger}
\end{align*}
Applying $\cU$ on the outputs of channel $W^+ $ $\forall u_2\in [q]$, we get 
\begin{align}\label{eq:bit node polarization}
 &   \cU W^+(u_2)\cU^{\dagger}\\
 &  =\frac{1}{q}\sum_{u_1\in [q]}\ketbra{u_1}{u_1}\otimes U_{u_1}^{\dagger} W_1(u_1+u_2) U_{u_1}\otimes W_2(u_2)\nonumber\\
 & = \frac{\mI}{q}\otimes W_{1}(u_2)\otimes W_{2}(u_2)\nonumber\\
 & = \frac{\mI}{q}\otimes [W_{1}\vnop W_2](u_2).
\end{align}
Thus, the channel $W^{+}$ is isometrically equivalent to channel $W_1\vnop W_2$. This means to understand the channel polarization in this setting, it is enough to analyze the check node and bit node channel recursions $W_1\cnop W_2$ and $W_1\vnop W_2$ where $W_1$ and $W_2$ are arbitrary symmetric $q$-ary PSCs.
\subsection{Channel Polarization}
Consider two symmetric $q$-ary PSCs $W_1$ and $W_2$. Since polarized channels $W^{+}$ and $W^{-}$ are constructed via invertible map $K_2$, we get 
\begin{align*}
    I(W^-)+I(W^+)=I(W_1)+I(W_2).
\end{align*}
Using \eqref{eq:check node polarization} and \eqref{eq:bit node polarization}, we get 
\begin{align*}
    I(W_1\cnop W_2)+ I(W_1\vnop W_2)= I(W_1)+I(W_2).
\end{align*}
This implies if the channels $W_1$ and $W_2$ are characterized by the normalized eigen list $\bmu^{(1)}$ and $\bmu^{(2)}$ then the following holds from Lemma~\ref{lem:Holevo_bitnode} we get 
\begin{align*}
    I(W_1\cnop W_2)=H(\bmu^{(1)})+H(\bmu^{(2)})- H(\bmu^{(1)}*\bmu^{(2})
\end{align*}
Similarly, consider two heralded symmetric PSCs $W_{H}^{(1)}$ and $W_{H}^{(2)}$ whose outputs are of the forms as below 
    \begin{align*}
        W_H^{(1)}(j) & =\sum_{x_1\in \cX_1}p_{x_1}^{(1)}W_{x_1}^{(1)}(j)\otimes \ketbra{x_1}{x_1}\\
        W_H^{(2)}(j) & =\sum_{x_2\in \cX_2}p_{x_2}^{(2)}W_{x_2}^{(2)}(j)\otimes \ketbra{x_2}{x_2}.
    \end{align*}
    In this case also we get 
    \begin{align*}
        I(W_{H}^{(1)}\cnop W_{H}^{(2)})+ I(W_H^{(1)}\vnop W_H^{(2)})= I(W_H^{(1)})+I(W_H^{(2)}). 
    \end{align*}
    Hence, using Lemma~\ref{lem: Holevo information for heralded PSC}, we get 
    \begin{align}\label{eq:heraled psc holevo information relation}
        I(W_{H}^{(1)}\cnop W_{H}^{(2)})= \mathbb{E}_{x_1\in \cX_1}[H(\bmu_{x_1}^{(1})]+\mathbb{E}_{x_2\in \cX_2}[H(\bmu_{x_2}^{(2})]- \mathbb{E}_{x_1\in \cX_1,x_2\in \cX_2}[H(\bmu^{(1)}_{x_1}*\bmu^{(2)}_{x_2})].
    \end{align}
    % \AM{whom should else we cite}
Since we assume $q$ to be prime, we get single level of polarization as there is no non trivial subgroup\cite{nasser2018polar}. In Lemma~\ref{lem:Holevo and fidelity boundary relations}, we showed when $I(W)\rightarrow 0$, we have $F(W)\rightarrow 1$ and when $I(W)\rightarrow\log q$, $F(W)\rightarrow 0$. Combining this with $\eqref{eq:heraled psc holevo information relation}$, we obtain the channel polarization for symmetric $q$-ary PSC based on the proof technique from \cite{csacsouglu2010entropy}.

Furthermore, if $q$ is composite, the channel polarization phenomenon still occurs but results in multilevel polarization, where channels converge to intermediate symmetric Holevo information determined by the subgroup structure of $\mathbb{Z}_q$. 
% We focus on the prime case here to maintain the standard perfect-vs-useless polarization behavior.
\section{Density Evolution}\label{Appendix:DE results}
 For binary symmetric CQ channels, DE was proposed in \cite{brandsen2022belief} based on a generalization of BPQM called paired measurement BPQM. The paired measurement compresses the decision information from check-node and bit-node combining into the first qubit while keeping reliability information in the second qubit. Applying DE to a long code, whose factor graph is a tree with sufficiently large depth, results in a threshold phenomenon that allows one to estimate asymptotic characteristics (such as noise threshold) of the code.
In~\cite{brandsen2022belief, mandal2024polarcodescqchannels}, this was applied to regular LDPC codes and polar codes on CQ channels with PM-BPQM decoding. In~\cite{mandal2024polarcodescqchannels,mandal2023belief}, DE was also used to design polar codes based on a targeted block error rate, while in \cite{piveteau2025efficient}, it was applied to turbo codes.
\subsection{DE Results for Polar codes on Symmetric PSC}
Consider length $N=2^n$ polar code where we want to find the information set $\cA$ for a targeted block error rate on a  symmetric $q$-ary PSC $W$ with circulant Gram matrix $G$ characterized by eigen list $\blambda$. This can be done using the DE based on the BPQM updated described in Lemmas~\ref{lem:check node},\ref{lem:bit node}. 
In the Appendix~\ref{sec: DE Algorithms}, we describe the detailed algorithms for Monte Carlo based DE for polar codes with arbitrary code length $N=2^n$ in a symmetric $q$-ary PSC with arbitrary eigen list $\blambda$.  The algorithms BitCombine (Algorithm~\ref{alg:bitcombine}) and CheckCombine (Algorithm~\ref{alg:check combine}) describe the bit node and check node operations for two collections (which we refer to as bags) of symmetric  $q$-ary PSCs. To obtain the polar coding design i.e. to obtain the information set $\cA$, we fix a target block error rate $\epsilon$ and rely on the non commutative union bound \cite{gao2015quantum}. More specifically, let $\perr(W^i_N)$ denote the optimal symbol error probability of channel $W^{i}_N$ using PGM as described in \ref{sec:pgm}, where $i\in [1,\dots, 2^n]$ is the index of polarized channel after $n$ rounds of polarization.

In Fig.~\ref{plot: polar design curve q=3}, we obtain the average channel error rate from the DE for each channel $W^{i}_N$ where we consider the symmetric PSC $W$ with $\blambda=[2.2,0.4,0.4]$. To compare this with respect to the different block lengths, we consider the normalized channel index $\frac{i}{N}$ of the channel $W^{i}_N$. Then, we sort the channel according to the ascending order of channel error rate to get the normalized channel rank and obtain the curve for each block length. We observe that, as increase the block length, we get sharp threshold behavior in terms channel error rate which emphasizes the channel polarization behavior for symmetric $q$-ary PSCs.  In Fig.~\ref{plot: polar design curve q=5}, we obtain the average channel error rate plot with respect to normalized channel rank for $q=5$ with $\blambda=[3,0.5,0.5,0.5,0.5]$.
In Fig.~\ref{plot: polar design curve q=3 arbitrary} and Fig.~\ref{plot: polar design curve q=5 arbitrary}, we obtain the channel error rate using DE on symmetric $q$-ary PSCs  whose eigen lists are $\blambda=[1.9,0.65,0.45]$ and $\blambda=[2,1.1,0.9,0.6,0.4]$ respectively. This provides evidence of DE framework for any arbitrary eigen list for fixed $q$.

In Fig.~\ref{plot: polar rate vs capacity q=5} and \ref{plot: polar rate vs capacity q=7}, we obtain the design rate for different block lengths $N$ with target block error rate $\epsilon=0.1$ in terms of $\lambda_0$ for channel $W$ with  eigen list of the form $\blambda=[\lambda_0,\frac{q-\lambda_0}{q-1},\dots,\frac{q-\lambda_0}{q-1}]$ for $q=5$ and $q=7$ respectively.
\begin{figure}[t]
    \centering
    \subfloat[$q=3$ and $\lambda_0=2.2$ \label{plot: polar design curve q=3}]{%
        \includegraphics[width=0.45\linewidth]{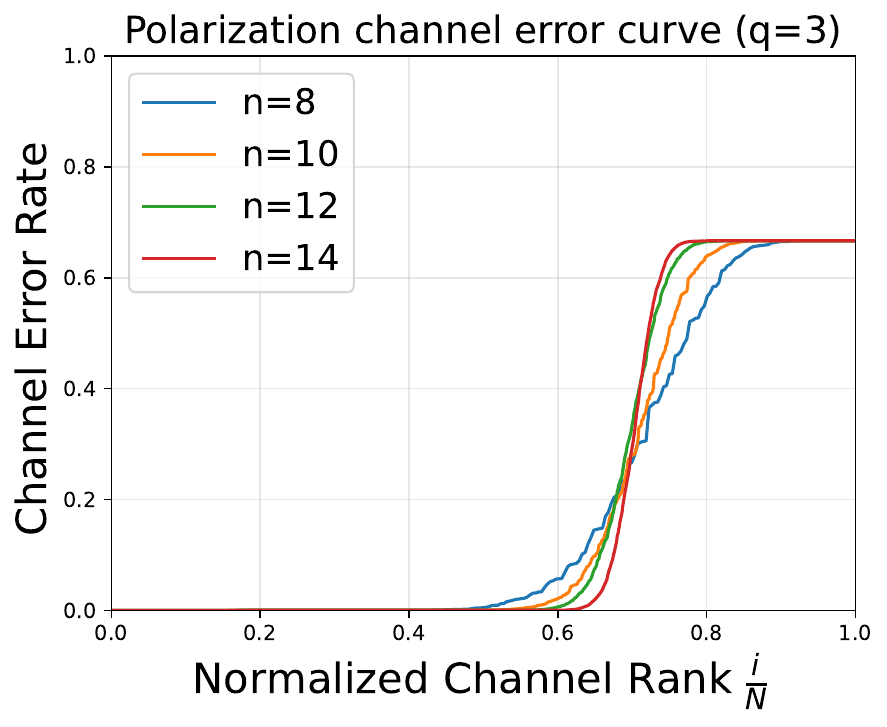}}
    \hfill
    \subfloat[$q=5$ and $\lambda_0=3$ \label{plot: polar design curve q=5}]{%
        \includegraphics[width=0.45\linewidth]{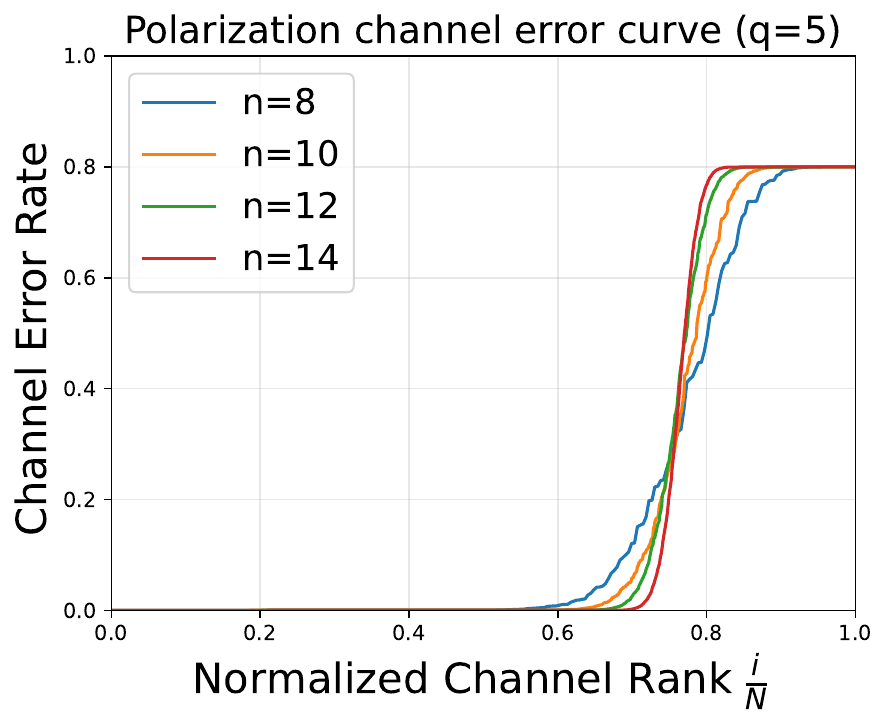}}
    
    \caption{Polar design curves for $q=3$ and $q=5$ for different block lengths $N=2^n$}
    % \label{}
\end{figure}

\begin{figure}[t]
    \centering
    \subfloat[$q=3$ and $\blambda={[1.9,0.65,0.45]}$ \label{plot: polar design curve q=3 arbitrary}]{%
        \includegraphics[width=0.45\linewidth]{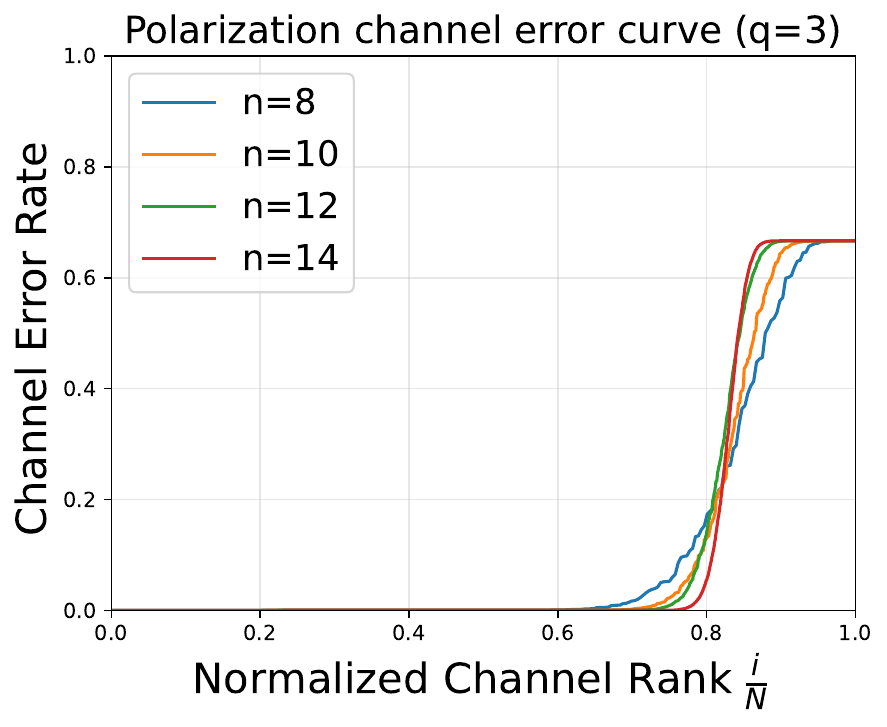}}
    \hfill
    \subfloat[$q=5$ and $\blambda={[2,1.1,0.9,0.6,0.4]}$ \label{plot: polar design curve q=5 arbitrary}]{%
        \includegraphics[width=0.45\linewidth]{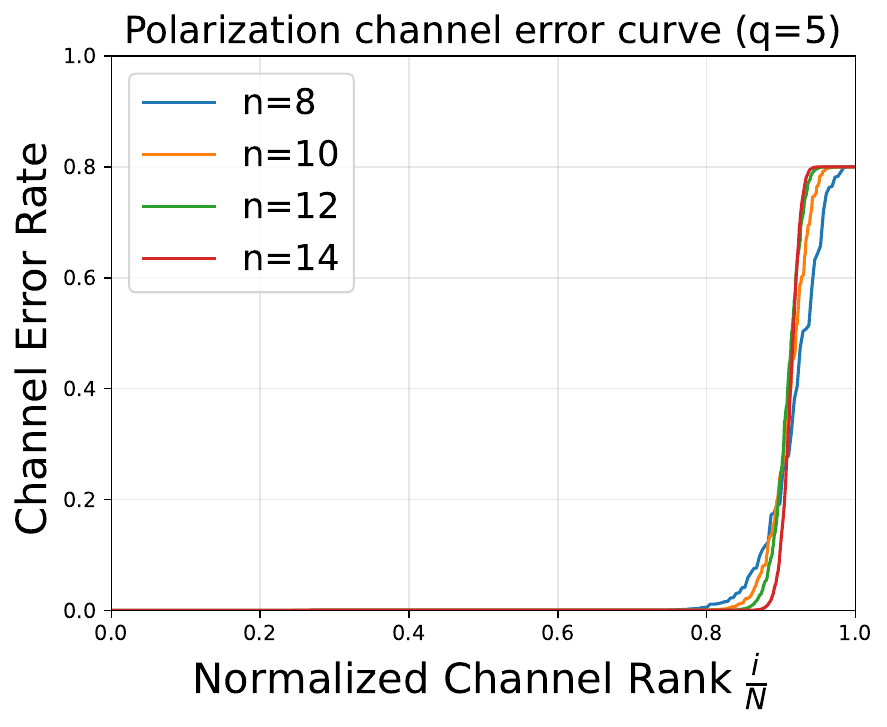}}
    
    \caption{Polar design curves for $q=3$ and $q=5$ for different block lengths $N=2^n$}
    % \label{}
\end{figure}
\begin{figure}[t]
    \centering
    \subfloat[$q=5$ \label{plot: polar rate vs capacity q=5}]{%
        \includegraphics[width=0.45\linewidth]{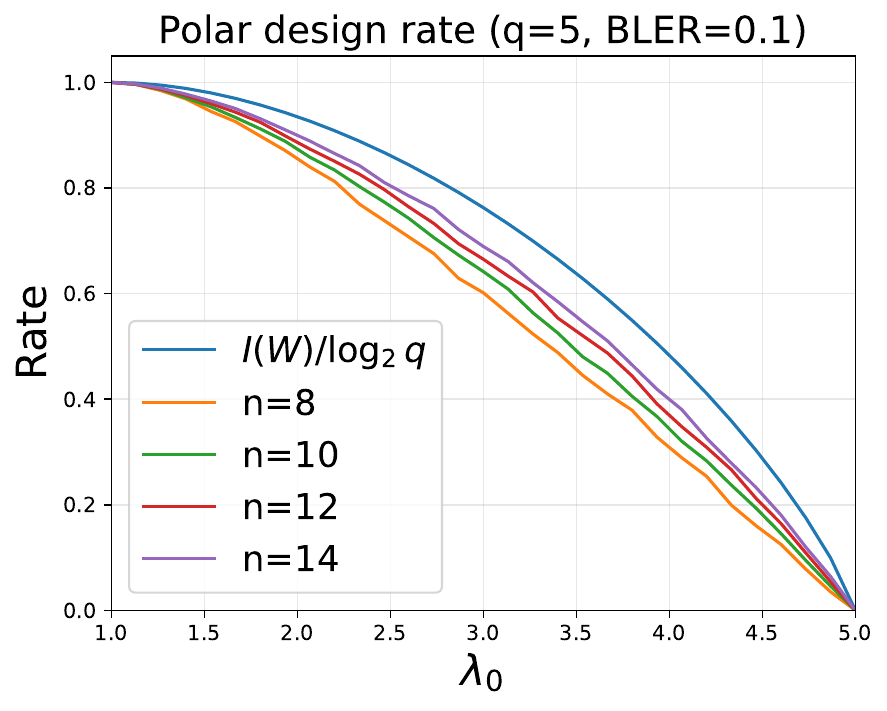}}
    \hfill
    \subfloat[$q=7$ \label{plot: polar rate vs capacity q=7}]{%
        \includegraphics[width=0.45\linewidth]{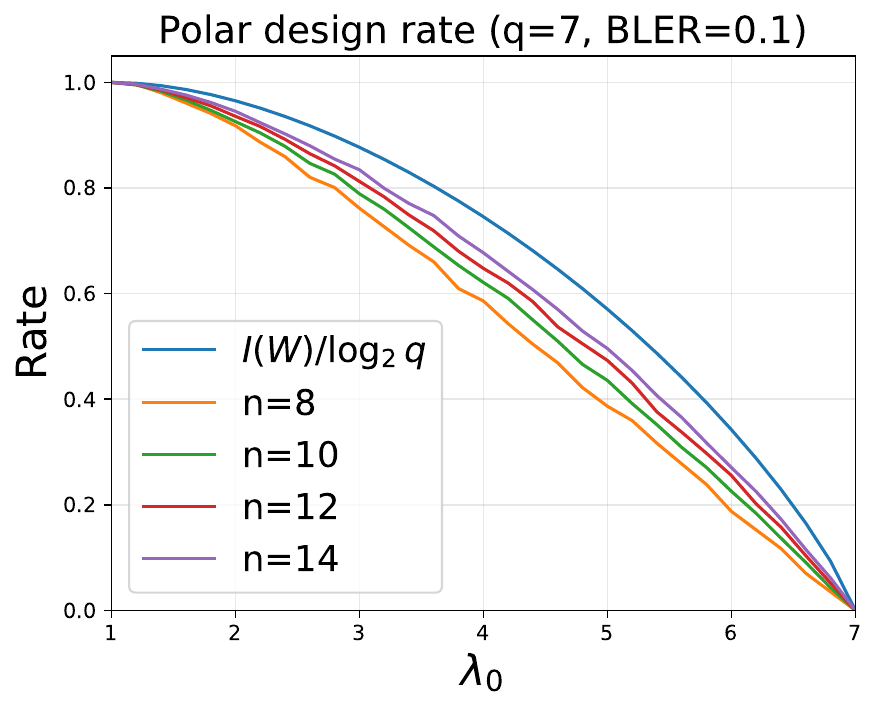}}
    
    \caption{Plot comparing rate $\frac{|\cA|}{N}$ of polar codes with different block lengths }
    % \label{}
\end{figure}
% \begin{figure}
  
%     \includegraphics[width=0.5\linewidth]{pdf_figs/design_curve_q3_specialcase.pdf}
%     \caption{}
%     \label{plot: polar design curve}
% \end{figure}
% \begin{figure}
%     \includegraphics[width=0.5\linewidth]{pdf_figs/design_curve_q5_specialcase.pdf}
%     \caption{}
%     \label{plot: polar design curve q5}
% \end{figure}
\section{DE Results for LDPC Codes on Symmetric PSC}

Although Gallager originally proposed Low-density parity-check (LDPC) codes and iterative decoding in 1960~\cite{gallager1962low}, the concepts remained largely dormant until the emergence of Turbo codes~\cite{berrou1993near} and the subsequent rediscovery of LDPC codes decades later~\cite{mackay2002good}. In a parallel development, Pearl introduced Belief Propagation (BP) in 1982~\cite{Pearlaaai82} as a method for efficiently computing exact marginals on tree-structured probability models via message passing. It was later established that Pearl's BP algorithm provides a general framework that encompasses both Turbo decoding and Gallager’s iterative decoding as specific instances~\cite{mceliece1998turbo,Kschischangjsac98}.

Density evolution has been widely employed to estimate the code performance in the context of BP decoding of LDPC codes over classical channel \cite{richardson2008modern}.
For $(dv,dc)$ regular LDPC code the idea of DE on symmetric PSCs $q$-ary can be roughly explained as follows. Let $W^{(t)}$ be the channel after $t$ rounds of BPQM starting with symmetric $q$-ary PSC $W$ on a factor graph which has bit degree $dv$ and check degree $dc$.  Of course, $W^{(t)}$ would be a heralded mixture of PSCs. For channel $W'$, let $(W')^{\cnop k}$ and $(W')^{\vnop k}$ denote the channel obtained via $k$ rounds of check node and bit channel combining with channel $W'$ itself. Then, to obtain the output channel for the next stage, first we apply check node combining $(dc-1)$ times on channel $W^{(t)}$ resulting channel $(W^{(t)})^{\cnop (dc-1)}$. This is followed by bit node combining $(dv-1)$ times to obtain the channel $[(W^{(t)})^{\cnop (dc-1)}]^{\vnop dv-1}$. Finally, we apply bit node combining on channel $W$ and $[(W^{(t)})^{\cnop (dc-1)}]^{\vnop dv-1}$, to obtain the output channel corresponding to $(t+1)\textsuperscript{th}$ stage $W^{(t+1)}=W\vnop [(W^{(t)})^{\cnop (dc-1)}]^{\vnop dv-1} $. DE essentially tracks the eigen list of these channels using the parameter update from Lemma~\ref{lem:check node} and \ref{lem:bit node}.  Once we know that eigen list associated with the heralded mixture of PSCs at each states, we can efficiently compute optimal symbol error probability and the symmetric Holevo information at each stage. Using this, we can obtain the threshold in terms of eigen list of the channel $W$. In this context, the threshold refers to channel $W_{th}$ such that $\lim_{t\to \infty}I(W_{th}^{(t)})=\log q$ and  for all symmetric $q$-ary PSC $W'$ with $I(W')< I(W)$, $\lim_{t\to \infty}I((W')^{(t)})=0$.

Consider the channel $W$ with eigen list $\blambda=[\lambda_0,\frac{q-\lambda_0}{q-1},\dots,\frac{q-\lambda_0}{q-1}]$. In Fig.~\ref{plot: ldpc threshold (3,4) q=3},\ref{plot: ldpc threshold (5,6) q=3}, \ref{plot: ldpc threshold (3,8) q=3}, \ref{plot: ldpc threshold (6,8) q=3}, we obtain the $\lambda_0$ threshold for $(3,4)$,$(5,6)$, $(3,8)$  and $(6,8)$ regular LDPC codes, with $q=3$.
\begin{figure}[ht]
    \centering
    \subfloat[ (3,4) LDPC code \label{plot: ldpc threshold (3,4) q=3}]{%
        \includegraphics[width=0.45\linewidth]{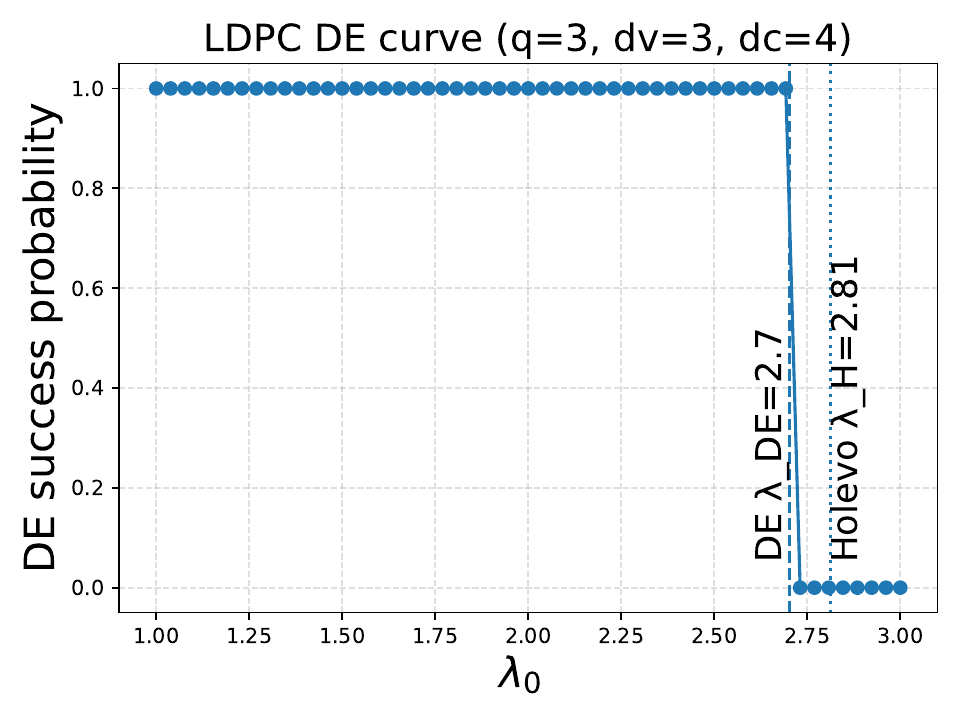}}
    \hfill
    \subfloat[ (5,6) LDPC code \label{plot: ldpc threshold (5,6) q=3}]{%
        \includegraphics[width=0.45\linewidth]{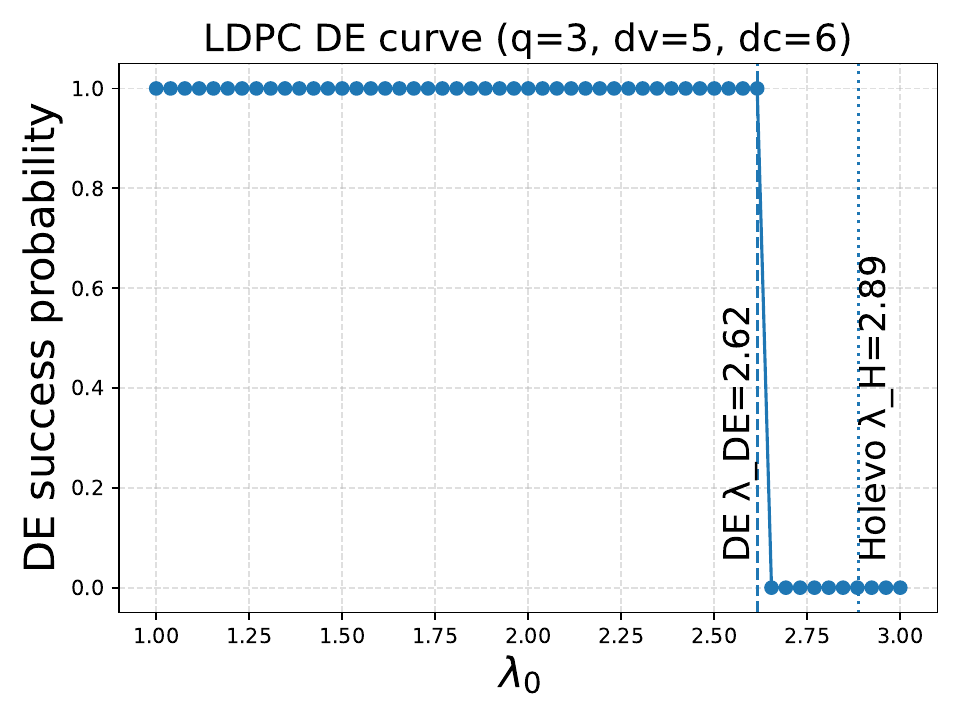}}
    
    \caption{LDPC threshold plots for $q=3$ }
    % \label{fig:both_curves}
\end{figure}

\begin{figure}[ht]
    \centering
    \subfloat[ (3,8) LDPC code \label{plot: ldpc threshold (3,8) q=3}]{%
        \includegraphics[width=0.45\linewidth]{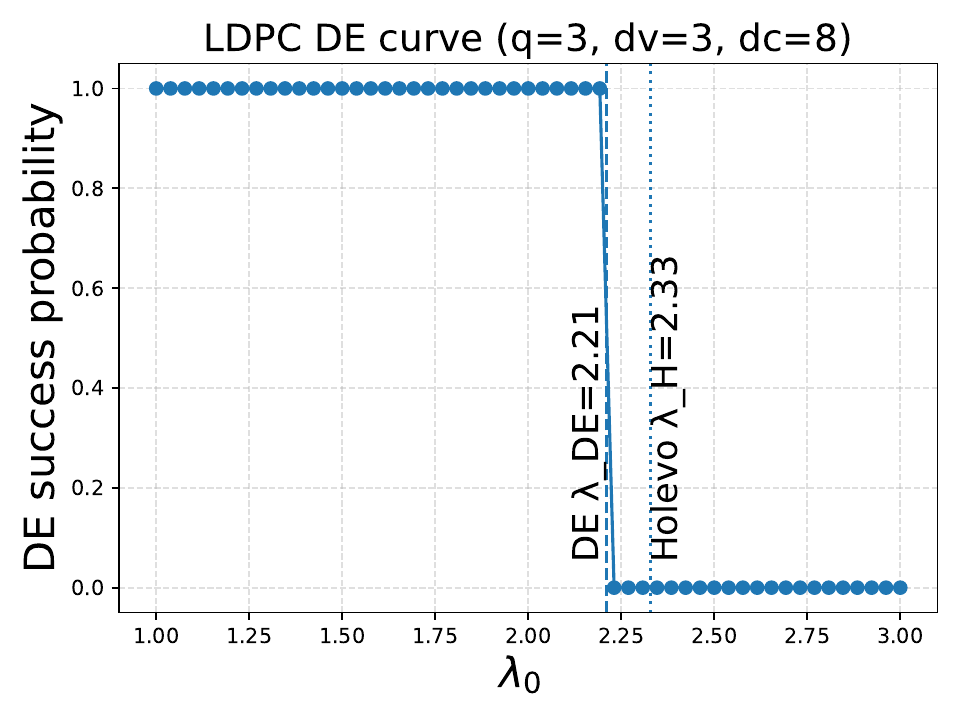}}
    \hfill
    \subfloat[ (6,8) LDPC code \label{plot: ldpc threshold (6,8) q=3}]{%
        \includegraphics[width=0.45\linewidth]{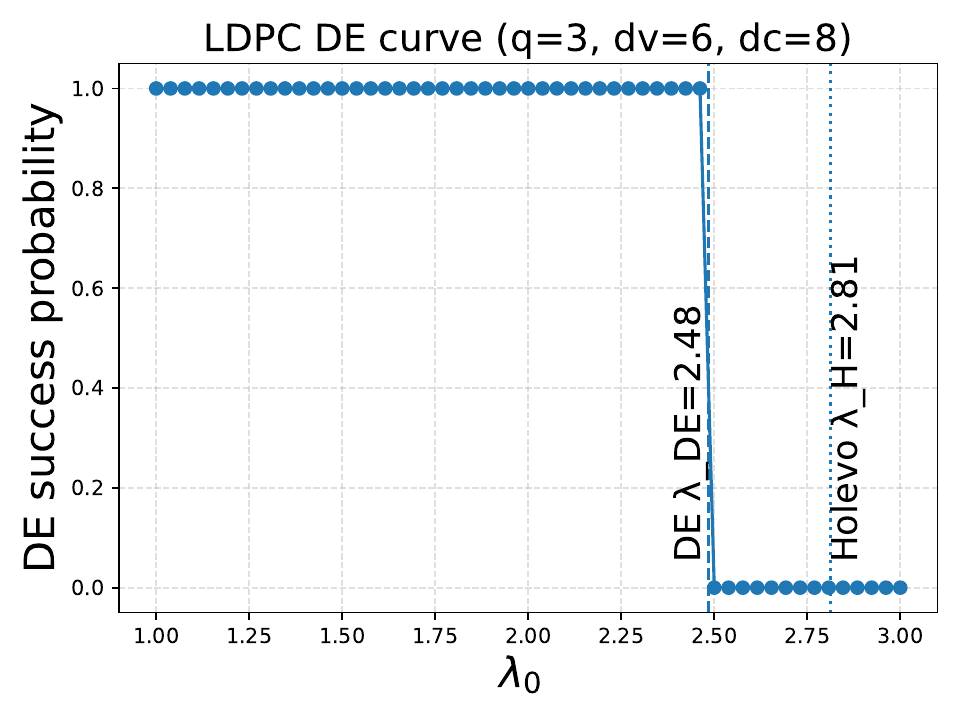}}
    
    \caption{LDPC threshold plots for $q=3$ }
    % \label{fig:both_curves}
\end{figure}

\twocolumn
\section{Algorithms for Density Evolution for Polar and LDPC Codes on Symmetric PSCs}\label{sec: DE Algorithms}

\begin{algorithm}[h!]
\caption{BitCombine$(B_1,B_2)$}\label{alg:bitcombine}
\begin{algorithmic}
    \Require Two Bags $B_1,B_2$ consisting of eigen lists \\
    Construct an empty bag $B^{\vnop}$\\
    Randomly permute the entries of $B_2$
    \While{$B_1,B_2$ not empty}\\
    \quad\qquad Sample $\blambda_1=[\lambdaa_0,\dots,\lambdaa_{q-1}]$ from $B_1$\\
    \quad\qquad Sample $\blambda_2=[\lambdab_0,\dots,\lambdab_{q-1}]$ from $B_2$\\
    \quad\qquad Apply bit node combining rule to obtain \\
    \quad\qquad $\lambda^{\vnop}_{j}\leftarrow\frac{1}{q}\sum_{k\in [q]}\lambdaa_{k}\lambdab_{j-k}$\\
    \quad\qquad Append $\blambda^{\vnop}=[\lambda^{\vnop}_0,\dots,\lambda^{\vnop}_{q-1}]$ to $B^{\vnop}$\\
    \quad\qquad Remove $\blambda_1$ from $B_1$ and $\blambda_2$ from $B_2$
    \EndWhile\\
    \Return $B^{\vnop}$
\end{algorithmic}
\end{algorithm}

\vfill

\begin{algorithm}[h!]
\caption{CheckCombine$(B_1,B_2)$}\label{alg:check combine}
\begin{algorithmic}
    \Require Two Bags $B_1,B_2$ consisting of eigen lists \\
    Construct an empty bag $B^{\cnop}$\\
    Randomly permute the entries of $B_2$
    \While{$B_1,B_2$ not empty}\\
    \quad\qquad Sample $\blambda_1=[\lambdaa_0,\dots,\lambdaa_{q-1}]$ from $B_1$\\
    \quad\qquad Sample $\blambda_2=[\lambdab_0,\dots,\lambdab_{q-1}]$ from $B_2$\\
    \quad\qquad Apply check node combining rule to obtain \\
    \quad\qquad $ p_{m}^{\cnop}  \leftarrow\frac{1}{q^2}\sum_{j\in [q]}\lambda_{(m+j)}^{(1)}\lambda_{-j}^{(2)}$\\
     \quad\qquad   $\lambda_{j}^{(\cnop,m)}  \leftarrow \frac{1}{qp_{m}^{\cnop}}\lambda_{(m+j)}^{(1)}\lambda_{-j}^{(2)}$\\
     \quad\qquad Sample $\blambda^{\cnop,m}=[\lambda_{0}^{\cnop,m},\dots,\lambda_{q-1}^{\cnop,m}]$ w.p. $p_{m}$\\
    \quad\qquad Append $\blambda^{\cnop,m}$ to $B^{\cnop}$\\
    \quad\qquad Remove $\blambda_1$ from $B_1$ and $\blambda_2$ from $B_2$
    \EndWhile\\
    \Return $B^{\cnop}$
\end{algorithmic}
\end{algorithm}

\vfill 

\begin{algorithm}[hb!]
\caption{DE for $q$-ary Polar Codes on a Symmetric PSC\label{algo:polar qudit de}}
\begin{algorithmic}[1]
      \Require Symmetric $q$-ary PSC $W$ characterized by Gram eigen list $\blambda$. Target block error rate $\epsilon$. \\
      Create bag of eigen lists, $B$, consisting of $M$ copies of $\blambda$\\
      Set $B=B^{1,0}$
      \For{$k$ in $\{1,\dots, n\}$}
      \For{$i$ in $\{1,\dots, 2^{k-1}\}$}\\
   \quad \qquad    \textbf{Compute check node update on bag $B^{i,k-1}$} \\
   \quad\qquad $B^{2i-1,k}\leftarrow \text{CheckCombine}(B^{i,k-1},B^{i,k-1})$\\
    \quad \qquad   \textbf{Compute bit node update on $B^{i,k-1}$} \\
   \quad\qquad $B^{2i,k}\leftarrow \text{BitCombine}(B^{i,k-1},B^{i,k-1})$

\EndFor
\EndFor
\For{$i$ in $\{1,\dots,N\}$}\\
\quad Compute empirical error probability $P_{\text{err}}(B^{i,N})$
\EndFor\\
Find the largest set $\cA$, satisfying $4\sum_{i\in \cA}\perr(B^{i,N})\leq \epsilon$\\
\Return $\cA$
    \end{algorithmic}
\end{algorithm}

\vfill

\newpage

\begin{algorithm}[!t]
\caption{ DE for $q$-ary $(d_v,d_c)$-regular LDPC Codes on a Symmetric PSC}
\label{alg:ldpc_de_qudit}
\begin{algorithmic}[1]
% \Require Regular degrees $(d_v,d_c)$, population size $M$, iterations (depth) $T$.
\Require Symmetric $q$-ary PSC $W$ with Gram eigen list $\blambda_0$.\\
Set number of samples $M$, tree depth $T$
% \Ensure Empirical error proxy $\{P_{\mathrm{err}}^{(t)}\}_{t=1}^{T}$ (and/or final population).
\State Create channel bag $B^{\mathrm{ch}}$ consisting of $M$ copies of $\blambda_0$.
\State Initialize variable-to-check bag $B^{(0)} \gets B^{\mathrm{ch}}$.

\For{$t=1$ to $T$}

  \Statex \hspace{-2mm}\textbf{(A) Check-node update: $(d_c-1)$-fold check combining}
  \State $B^{\cnop} \gets B^{(t-1)}$
  \For{$r=2$ to $d_c-1$}
     % \State $B' \gets \textsc{RandomPermute}\!\left(B^{(t-1)}\right)$
     \State $B^{\cnop} \gets \text{CheckCombine}(B^{\cnop}, B^{(t-1)})$
  \EndFor

  \Statex \hspace{-2mm}\textbf{(B) Variable-node update: $(d_v-1)$-fold bit combining}
  \State $B^{\mathrm{ext}} \gets B^{\cnop}$
  \For{$r=2$ to $d_v-1$}
     % \State $B' \gets \textsc{RandomPermute}\!\left(B^{\cnop}\right)$
     \State $B^{\mathrm{ext}} \gets \text{BitCombine}(B^{\mathrm{ext}}, B^{\cnop})$
  \EndFor

  \Statex \hspace{-2mm}\textbf{(C) With channel message (bit combining with $\blambda_0$)}
  % \State $B' \gets \textsc{RandomPermute}\!\left(B^{\mathrm{ext}}\right)$
  \State $B^{(t)} \gets \text{BitCombine}(B^{\mathrm{ch}}, B^{\mathrm{ext}})$

  \Statex \hspace{-2mm}\textbf{(D) Track PGM error } %(optional; matches the Python DE script)}
  \State $P_{\mathrm{err}}^{(t)} \gets 1 - \frac{1}{M}\sum_{\blambda \in B^{(t)}} \psuc(\blambda)$
  % \Comment{e.g., PGM proxy: $\psuc(\blambda)=\big(\frac{1}{q}\sum_{j\in[q]}\sqrt{\lambda_j}\big)^2$}

\EndFor

\State \Return $\{P_{\mathrm{err}}^{(t)}, B^{(t)}\}_{t=1}^{T}$ 
\end{algorithmic}
\end{algorithm}

\vfill
\null 
\clearpage
\onecolumn
% \section{Generalization to channels with abelian group symmetry}
% \input{bpqm for abelian symmetry}
\end{appendices}

\end{document}